\documentclass[journal,onecolumn]{IEEEtran}
\usepackage{amsmath,amsfonts,amssymb,mathrsfs}
\usepackage{array}
\usepackage{cite}
\usepackage{url}
\newtheorem{theorem}{Theorem}
\newtheorem{lemma}{Lemma}
\newtheorem{proposition}{Proposition}
\newtheorem{corollary}{Corollary}

\newtheorem{remark}{Remark}
\newenvironment{proof}[1][Proof]{\noindent\emph{#1: }\ignorespaces}{\hfill$\IEEEQED$\par}

\DeclareMathOperator{\wt}{wt}
\DeclareMathOperator{\Supp}{Supp}
\DeclareMathOperator{\Aff}{Aff}
\DeclareMathOperator{\Span}{span}
\DeclareMathOperator{\RM}{RM}

\DeclareMathOperator{\ev}{ev}

\newcommand{\Fq}{\mathbb F_q}
\newcommand{\Fqm}{\mathbb F_{q^m}}
\newcommand{\PP}{\mathbb P}

\newcommand{\calG}{\mathcal G}

\newcommand{\calI}{\mathcal I}

\newcommand{\calM}{\mathcal M}
\newcommand{\calR}{\mathcal R}
\newcommand{\bfzero}{\mathbf 0}

\begin{document}

\title{Intermediate Constacyclic Codes and Scalar-Residue Reed--Muller Layers}

\author{Yaoran~Yang,~\IEEEmembership{Member,~IEEE,}
        and~Yutong~Zhang,~\IEEEmembership{Member,~IEEE}%
\thanks{Manuscript received Month Day, 2026; revised Month Day, 2026.}
\thanks{Yaoran Yang and Yutong Zhang are with the School of Mathematics,
Sichuan University, Chengdu 610065, China
(e-mail: yangyaoran@stu.scu.edu.cn; yutongzhang@stu.scu.edu.cn).}
\thanks{Corresponding author: Yutong Zhang.}}

\markboth{IEEE Transactions on Information Theory,~Vol.~XX, No.~XX, Month~2026}%
{Yang and Zhang: Constacyclic Codes and Reed--Muller Residue Layers}

\maketitle

\begin{abstract}
A 2024 paper of Sun, Ding and Wang introduced a second class of
constacyclic codes over finite fields, denoted $C(q,m,r,\ell)$, with length
$(q^m-1)/r$, where $r\mid(q-1)$ and the defining monomials have total
$q$-ary degree congruent to $r-1$ modulo $r$.  In the non-projective
intermediate range $2<r<q-1$ the paper gave a sharp-looking upper bound
and a BCH-type lower bound, and left the minimum distance open.  We prove that the upper bound is the exact minimum distance in the
non-terminal part of the intermediate range, and we record the terminal
endpoint separately.  More precisely, if
$\ell=(q-1)a+b<(q-1)m-1$, $0\le b\le q-2$, and
$b\equiv r-1\pmod r$, then, for every prime power $q$, every divisor
$r$ of $q-1$ with $2<r<q-1$, and every $m\ge2$,
\[
  d(C(q,m,r,\ell))=
  \begin{cases}
  \displaystyle \frac{q-1}{r}(q-b+1)q^{m-a-2},&0\le a\le m-2,\\[1mm]
  \displaystyle \frac{q-b+r-2}{r},&a=m-1.
  \end{cases}
\]
In the second line the condition $\ell<(q-1)m-1$ forces
$b\le q-3$.  The first line settles the open problem of Sun, Ding and
Wang in the non-terminal range; the second line is the terminal case
known from their terminal estimate.  We also determine the minimum affine
support of the top residue-matched non-terminal scalar-residue layers of
generalized Reed--Muller codes.  The resulting dichotomy says that, when
the top degree residue is fixed, the first Reed--Muller weight survives
exactly for residue classes $0$ and $1$, while every other residue-matched
layer starts at the second Reed--Muller weight.
\end{abstract}

\begin{IEEEkeywords}
Affine Reed--Muller codes, constacyclic codes, finite fields, minimum
distance, projective Reed--Muller codes, scalar orbits, residue layers.
\end{IEEEkeywords}

\section{Introduction}
\IEEEPARstart{C}{onstacyclic} codes form one of the basic algebraic
families between cyclic codes and general linear codes.  Let $q$ be a
prime power, let $r\mid(q-1)$, let $m\ge2$, and put
\begin{equation}
        n=\frac{q^m-1}{r}.                                      \label{eq:n}
\end{equation}
Sun, Ding and Wang~\cite{SunDingWang2024} constructed, among other
families, a class of
$\lambda$-constacyclic codes of length $n$ over $\Fq$, denoted
$C(q,m,r,\ell)$, for degrees
\begin{equation}
        r-1\le \ell < (q-1)m-1,
        \qquad       \ell\equiv r-1\pmod r .                   \label{eq:ellrange}
\end{equation}
The construction is naturally described by evaluating polynomials whose
reduced monomials have total degree at most $\ell$ and total degree
congruent to $r-1$ modulo $r$.  For $r=q-1$ this family is essentially a
projective Reed--Muller family; for $r=2$ several coincidences also make
the minimum distance accessible.  The genuinely intermediate regime is
\begin{equation}
        2<r<q-1.                                                \label{eq:intermediate}
\end{equation}
For background on constacyclic codes over finite fields and standard
coding-theoretic notation, see
\cite{HuffmanPless,MacWilliamsSloane,ChenFanLinLiu2012,DinhLopezPermouth2004,SharmaRani2016,LidlNiederreiter}.

In that regime Sun, Ding and Wang proved the two-sided estimate
\begin{IEEEeqnarray}{rCl}
 d(C(q,m,r,\ell))
 &\ge& \displaystyle \frac{(q-b)q^{m-a-1}-2}{r}+1,\nonumber\\
 d(C(q,m,r,\ell))
 &\le& \displaystyle \frac{q-1}{r}(q-b+1)q^{m-a-2}.               \label{eq:SDWbound}
\end{IEEEeqnarray}
when $\ell=(q-1)a+b$, $0\le a\le m-2$, $0\le b\le q-2$, and
$b\equiv r-1\pmod r$.  They asked to determine the minimum distance, or
at least improve the lower bound, for
\begin{equation}
        \ell=rL+r-1,
        \qquad 1\le L\le \frac{q-1}{r}m-3 .                     \label{eq:openrange}
\end{equation}
The upper bound in \eqref{eq:SDWbound} has the form of a projective
pencil weight.  The lower bound, by contrast, is essentially BCH-like and
has the scale of the first affine Reed--Muller weight divided by $r$.
This paper proves that the upper bound is exact throughout the
non-terminal part of \eqref{eq:openrange}; the possible terminal endpoint
is governed by the separate terminal formula recorded below.  Thus, in
the non-terminal range, the right answer is the second affine
Reed--Muller weight divided by the scalar orbit length $r$.

The result may be summarized as follows.  Write every admissible degree
in the unique form
\begin{equation}
        \ell=(q-1)a+b,
        \,\quad 0\le b\le q-2,
        \,\quad b\equiv r-1\pmod r.                             \label{eq:abform}
\end{equation}
Then the non-terminal distances are
\begin{IEEEeqnarray}{rCl}
        d(C(q,m,r,\ell))
        &=& \displaystyle \frac{q-1}{r}(q-b+1)q^{m-a-2},\nonumber\\
        && 0\le a\le m-2 .                                      \label{eq:mainintro}
\end{IEEEeqnarray}
and the terminal distances are
\begin{equation}
        d(C(q,m,r,\ell))=\frac{q-b+r-2}{r},
        \qquad a=m-1.                                           \label{eq:terminalintro}
\end{equation}
For parameters in \eqref{eq:openrange}, the terminal case is either
outside the stated open range or covered by the existing BCH argument;
the new content is \eqref{eq:mainintro}.  In particular, the gap in
\eqref{eq:SDWbound} closes without changing the upper construction.

A recent square-root-like lower-bound construction of Chen, Sun, Xie,
Chen and Ding targets different constacyclic and negacyclic defining sets
and asks for minimum-distance improvements for those new classes
\cite{ChenSunXieChenDing2024}.  It does not determine the scalar-residue
subcode \(C(q,m,r,\ell)\) treated here.

\subsection{Idea of the Proof}
Let $\calM(q,m,r,\ell)$ be the reduced monomial space generated by
monomials
\begin{IEEEeqnarray}{c}
        X_1^{i_1}\cdots X_m^{i_m},\qquad
        0\le i_j\le q-1,\nonumber\\
        \sum_{j=1}^m i_j\le\ell,
        \qquad
        \sum_{j=1}^m i_j\equiv r-1\pmod r .                    \label{eq:Mconditions}
\end{IEEEeqnarray}
The crucial point is the scalar law
\begin{equation}
        f(\alpha P)=\alpha^{r-1}f(P)=\alpha^{-1}f(P)
        \qquad(\alpha^r=1),                                    \label{eq:scalarlawintro}
\end{equation}
valid for all $f\in\calM(q,m,r,\ell)$ and $P\in\Fq^m$.  Hence
$\Supp(f)$ is a union of the nonzero scalar orbits of the group
\begin{equation}
        \mu_r=\{\alpha\in\Fq^\ast:\alpha^r=1\},
        \qquad |\mu_r|=r,                                      \label{eq:mur}
\end{equation}
and the length-$n$ codeword has weight
\begin{equation}
        \wt(c_f)= \frac{|\Supp(f)|}{r}.                         \label{eq:wtdivideintro}
\end{equation}
Since $\calM(q,m,r,\ell)$ is contained in the affine generalized
Reed--Muller space of degree $\ell$, any nonzero $f$ has affine support
at least equal to the first Reed--Muller weight unless it is forced into
a higher weight stratum.  The first affine Reed--Muller supports of
order $\ell=(q-1)a+b$ have cardinality
\begin{equation}
        W_1=(q-b)q^{m-a-1},                                    \label{eq:W1intro}
\end{equation}
and are affine cylinders over an affine line with exactly $b$ deleted
parallel hyperplanes.  Such a support cannot satisfy
\eqref{eq:scalarlawintro} when $2<r<q-1$ and
$b\equiv r-1\pmod r$: scalar invariance forces the deleted set on the
line to contain $0$ and then to be a union of $\mu_r$-orbits, so its
cardinality is $1$ modulo $r$, not $r-1$ modulo $r$.  Therefore the
support must have size at least the second affine Reed--Muller weight
\begin{equation}
        W_2=(q-1)(q-b+1)q^{m-a-2}.                             \label{eq:W2intro}
\end{equation}
The upper construction is the homogeneous pencil
\begin{equation}
        F_{a,b}(X)=
        \prod_{i=1}^a (1-X_i^{q-1})
        \prod_{j=1}^b (X_{a+1}-\theta_jX_{a+2}),               \label{eq:pencilintro}
\end{equation}
where $\theta_1,\ldots,\theta_b$ are distinct elements of $\Fq$.  Every
monomial of \eqref{eq:pencilintro} has total degree congruent to $b$,
hence to $r-1$, modulo $r$.  Its support consists of
$q^{m-a-2}$ copies of $q+1-b$ projective directions, each contributing
$q-1$ nonzero affine points, and therefore has cardinality $W_2$.
Dividing by $r$ gives \eqref{eq:mainintro}.

\subsection{Organization}
Section~II fixes the code model and recalls the affine Reed--Muller facts
used in the sequel.  Section~III proves the orbit-weight identity.
Section~IV contains the obstruction to first-weight supports.  Section~V
gives the sharp construction and states the exact distance.  Section~VI
records the terminal case.  Sections~VII--IX give refinements and checks.
Section~X proves the stronger residue-layer dichotomy and its quotient-code
consequence.  The remaining sections record examples, boundary cases, and
numerical consistency checks.  No appendix is used; all details are included
in the body.

\section{Preliminaries}
\subsection{Finite-Field and Evaluation Notation}
Let $\Fq$ be the field of order $q$.  A polynomial in
$\Fq[X_1,\ldots,X_m]$ will always be represented by its reduced residue
modulo the ideal
\begin{equation}
        \calI_q=(X_1^q-X_1,\ldots,X_m^q-X_m).                  \label{eq:ideal}
\end{equation}
Thus every function $\Fq^m\to\Fq$ is represented uniquely as
\begin{equation}
        f(X)=\sum_{0\le i_1,\ldots,i_m\le q-1}
        c_i X_1^{i_1}\cdots X_m^{i_m},                         \label{eq:reduced}
\end{equation}
where the exponent $q-1$ is allowed.  Its reduced degree is
\begin{equation}
        \deg_q(f)=\max\Bigl\{\sum_{j=1}^m i_j:c_i\ne0\Bigr\}.   \label{eq:reduceddegree}
\end{equation}
For $0\le s\le m(q-1)$ the affine generalized Reed--Muller code is
\begin{equation}
        \RM_q(s,m)=\{(f(P))_{P\in\Fq^m}: \deg_q(f)\le s\}.      \label{eq:RMdef}
\end{equation}
If
\begin{equation}
        s=(q-1)a+b,
        \quad 0\le a\le m,
        \quad 0\le b\le q-2,                                  \label{eq:sab}
\end{equation}
with the usual convention that $a=m$ only when $b=0$, then the minimum
distance of $\RM_q(s,m)$ is
\begin{equation}
        d_1(s,m)=(q-b)q^{m-a-1}                               \label{eq:RMfirst}
\end{equation}
for $a\le m-1$.  In the non-terminal range $a\le m-2$ and
$2\le b\le q-2$, the next nonzero weight after the minimum weight is
\begin{equation}
        d_2(s,m)=(q-1)(q-b+1)q^{m-a-2}.                       \label{eq:RMsecond}
\end{equation}
The first-weight classification and the second-weight formula are standard;
we use them in the forms of
\cite{DelsarteGoethalsMacWilliams1970,KasamiTokuraAzumi1976,KasamiTokura1970,HeijnenPellikaan1998,Rolland2010,Leducq2013}.
The cases $b=0$ and $b=1$ require separate statements.  In the present
paper $b\equiv r-1\pmod r$ and $r>2$, so
\begin{equation}
        2\le b\le q-2,                                        \label{eq:bge2}
\end{equation}
and \eqref{eq:RMsecond} is the relevant formula.

\subsection{The Second Class of Constacyclic Codes}
Fix $r\mid(q-1)$ and $r>1$.  Let $\lambda\in\Fq^\ast$ have order $r$.
Choose a primitive element $\beta$ of $\Fqm$ satisfying $\beta^n=\lambda$
with $n=(q^m-1)/r$, and let $M$ be the companion matrix for multiplication
by $\beta$ in a fixed $\Fq$-basis of $\Fqm$.  If
$e=(1,0,\ldots,0)$, then
\begin{equation}
        \Fq^m\setminus\{\bfzero\}
        =\{eM^i:0\le i\le q^m-2\},                            \label{eq:affineorbit}
\end{equation}
and
\begin{equation}
        M^n=\lambda I_m.                                      \label{eq:Mn}
\end{equation}
For an admissible $\ell$ in \eqref{eq:ellrange}, define
\begin{IEEEeqnarray}{rCl}
\calM(q,m,r,\ell)
 &=&\Span_{\Fq}\{X_1^{i_1}\cdots X_m^{i_m}:\nonumber\\
 && 0\le i_j\le q-1,
        \quad i_1+\cdots+i_m\le \ell,\nonumber\\
 && i_1+\cdots+i_m\equiv r-1\pmod r\}.                       \label{eq:Mdef}
\end{IEEEeqnarray}
The evaluation model of Sun--Ding--Wang~\cite{SunDingWang2024} is
\begin{IEEEeqnarray}{rCl}
        \calG(q,m,r,\ell)
        &=&\{c_f=(f(e),f(eM),\ldots,f(eM^{n-1})):\nonumber\\
        &&\hspace{8mm} f\in\calM(q,m,r,\ell)\}.                 \label{eq:Gdef}
\end{IEEEeqnarray}
They prove that
\begin{equation}
        \calG(q,m,r,\ell)=C(q,m,r,\ell).                      \label{eq:GCequalsC}
\end{equation}
Consequently, the problem of determining $d(C(q,m,r,\ell))$ is equivalent
to a restricted affine support problem.

\begin{remark}
The use of $eM^i$ in \eqref{eq:Gdef} chooses one representative from each
orbit of multiplication by $\lambda$ on
$\Fq^m\setminus\{\bfzero\}$.  Since $\lambda$ has order $r$, the full
nonzero affine space decomposes as
\begin{equation}
        \Fq^m\setminus\{\bfzero\}
        =\coprod_{i=0}^{n-1}\{\lambda^j eM^i:0\le j\le r-1\}.  \label{eq:orbitdecomp}
\end{equation}
This decomposition is the bridge between the constacyclic length $n$ and
the affine Reed--Muller length $q^m$.
\end{remark}

\subsection{Earlier Bounds in the Same Parameters}
Put
\begin{equation}
        \ell=(q-1)a+b,
        \qquad 0\le b\le q-2,
        \qquad b\equiv r-1\pmod r .                           \label{eq:ellab}
\end{equation}
The previously known estimate in the intermediate non-terminal range was
\cite[Th. 47]{SunDingWang2024}
\begin{IEEEeqnarray}{rCl}
        d(C(q,m,r,\ell))
        &\ge& \displaystyle \frac{(q-b)q^{m-a-1}-2}{r}+1,\nonumber\\
        d(C(q,m,r,\ell))
        &\le& \displaystyle \frac{q-1}{r}(q-b+1)q^{m-a-2}.       \label{eq:oldbound}
\end{IEEEeqnarray}
The quotient between the two main terms is approximately
\begin{equation}
        \frac{(q-1)(q-b+1)}{q(q-b)}.                            \label{eq:ratio}
\end{equation}
This number is close to one for small $b$, but the exact distance is a
structural question: the first affine Reed--Muller layer either survives
the congruence restriction or it does not.  The rest of the paper proves
that it never survives when $2<r<q-1$.

\section{Scalar Orbits and Weight Transfer}
\subsection{Homogeneity Modulo $r$}
Let
\begin{equation}
        \mu_r=\{\alpha\in\Fq^\ast:\alpha^r=1\}.                \label{eq:mur2}
\end{equation}
For $f\in\calM(q,m,r,\ell)$ and $\alpha\in\mu_r$, every monomial in
$f$ satisfies
\begin{IEEEeqnarray}{rCl}
        X^i(\alpha P)&=&\alpha^{|i|}X^i(P),\nonumber\\
        |i|&=&i_1+\cdots+i_m,
        \quad |i|\equiv r-1\pmod r .                           \label{eq:monlaw}
\end{IEEEeqnarray}
so
\begin{equation}
        f(\alpha P)=\alpha^{r-1}f(P)=\alpha^{-1}f(P).          \label{eq:scalar-law}
\end{equation}
In particular,
\begin{equation}
        f(P)=0 \quad\Longleftrightarrow\quad f(\alpha P)=0.    \label{eq:zero-orbits}
\end{equation}
Because all monomials have positive degree, $f(\bfzero)=0$.  Hence
\begin{equation}
        \Supp(f)\subseteq \Fq^m\setminus\{\bfzero\}.            \label{eq:zerosat0}
\end{equation}

\begin{lemma}[Orbit weight identity]\label{lem:orbitweight}
For every $f\in\calM(q,m,r,\ell)$,
\begin{equation}
        \wt(c_f)=\frac{|\Supp(f)|}{r},                          \label{eq:orbitweight}
\end{equation}
where $c_f$ is the length-$n$ word in \eqref{eq:Gdef} and
$\Supp(f)=\{P\in\Fq^m:f(P)\ne0\}$.
\end{lemma}

\begin{proof}
The nonzero affine points are partitioned as in \eqref{eq:orbitdecomp}.
By \eqref{eq:zero-orbits}, either the entire orbit
$\{\lambda^j eM^i:0\le j<r\}$ is contained in $\Supp(f)$ or the entire
orbit is disjoint from $\Supp(f)$.  The coordinate $f(eM^i)$ of $c_f$
is nonzero exactly in the first case.  Thus every nonzero coordinate of
$c_f$ accounts for exactly $r$ affine support points, and
\eqref{eq:orbitweight} follows.
\end{proof}

The identity \eqref{eq:orbitweight} will be used in the following form.
Let
\begin{equation}
        \delta_{\rm aff}(q,m,r,\ell)=
        \min\{|\Supp(f)|:0\ne f\in\calM(q,m,r,\ell)\}.          \label{eq:deltaaff}
\end{equation}
Then
\begin{equation}
        d(C(q,m,r,\ell))=\frac{1}{r}\delta_{\rm aff}(q,m,r,\ell).\label{eq:d-delta}
\end{equation}
Since
\begin{IEEEeqnarray}{rCl}
        \calM(q,m,r,\ell)
        &\subseteq&
        \{f\in\Fq[X_1,\ldots,X_m]/\calI_q:\nonumber\\
        && \deg_q(f)\le\ell,
        \quad f(0)=0\}.                                        \label{eq:MinsideRM}
\end{IEEEeqnarray}
we always have
\begin{equation}
        \delta_{\rm aff}(q,m,r,\ell)\ge d_1(\ell,m).           \label{eq:trivialRM}
\end{equation}
The task is to improve \eqref{eq:trivialRM} by one weight level.

\subsection{A Useful Reformulation}
Let $G=\Fq^\ast/\mu_r$.  For a nonzero point $P$ write
\begin{equation}
        [P]_r=\{\alpha P:\alpha\in\mu_r\}.                     \label{eq:rorbit}
\end{equation}
The support of any nonzero word is a set of such $r$-orbits.  If
$\pi_r:\Fq^m\setminus\{0\}\to(\Fq^m\setminus\{0\})/\mu_r$ denotes the
quotient map, then
\begin{equation}
        \wt(c_f)=|\pi_r(\Supp(f))|.                            \label{eq:quotientweight}
\end{equation}
For a homogeneous polynomial $h$ of degree $d\equiv r-1\pmod r$,
\begin{equation}
        h(\alpha P)=\alpha^{-1}h(P),\qquad \alpha\in\mu_r,       \label{eq:homquot}
\end{equation}
and the same quotient formula applies.  The nonhomogeneous factors
$1-X_i^{q-1}$ in \eqref{eq:pencilintro} are also invariant under
$\Fq^\ast$, because
\begin{equation}
        (\alpha X_i)^{q-1}=X_i^{q-1}
        \qquad(\alpha\in\Fq^\ast).                             \label{eq:indicatorinvariant}
\end{equation}
Thus the construction in \eqref{eq:pencilintro} is perfectly compatible
with \eqref{eq:scalar-law}.

\section{Excluding the First Reed--Muller Layer}
\subsection{Minimum Supports of Affine Reed--Muller Codes}
We use the standard classification of minimum weight words in generalized
affine Reed--Muller codes.  If
$s=(q-1)a+b$, $0\le a\le m-1$, $0\le b\le q-2$, and
$0\ne f\in\RM_q(s,m)$ has support size $d_1(s,m)$, then after an affine
change of variables
\begin{equation}
        f(X)=\gamma
        \prod_{i=1}^{a}\bigl(1-(L_i(X)-u_i)^{q-1}\bigr)
        \prod_{j=1}^{b}(L_{a+1}(X)-v_j),                      \label{eq:RMminshape}
\end{equation}
where $\gamma\in\Fq^\ast$, the affine forms
$L_1,\ldots,L_{a+1}$ have independent linear parts, and
$v_1,\ldots,v_b$ are distinct elements of $\Fq$.  Its support is
\begin{IEEEeqnarray}{rCl}
S &=& \{P\in\Fq^m:L_i(P)=u_i\ (1\le i\le a),\nonumber\\
  &&\hspace{7mm} L_{a+1}(P)\notin\{v_1,\ldots,v_b\}\}\nonumber\\
  &=& A\setminus\bigcup_{j=1}^{b} H_j .                        \label{eq:RMminsupp}
\end{IEEEeqnarray}
where $A$ is an affine subspace of dimension
\begin{equation}
        t=m-a,                                                   \label{eq:tdef}
\end{equation}
and $H_1,\ldots,H_b$ are parallel affine hyperplanes of $A$.  Hence
\begin{equation}
        |S|=(q-b)q^{t-1}=d_1(s,m).                              \label{eq:Ssize}
\end{equation}

In our parameters $b\ge2$ by \eqref{eq:bge2}.  If $a\le m-2$, then
$t\ge2$, which is exactly the range where the second-weight value
\eqref{eq:RMsecond} is distinct from the first.

\subsection{Scalar-Invariant First Supports}
Suppose, for contradiction, that
\begin{equation}
        0\ne f\in\calM(q,m,r,\ell),
        \qquad |\Supp(f)|=d_1(\ell,m).                         \label{eq:firstcontra}
\end{equation}
Put $S=\Supp(f)$.  Then $S$ has the form \eqref{eq:RMminsupp}.  Since
$f(0)=0$, one has
\begin{equation}
        0\notin S.                                             \label{eq:0notS}
\end{equation}
Since $S$ is a union of scalar orbits under $\mu_r$,
\begin{equation}
        \alpha S=S,
        \qquad \alpha\in\mu_r.                                 \label{eq:S-invariant}
\end{equation}
We show that \eqref{eq:RMminsupp}, \eqref{eq:0notS}, and
\eqref{eq:S-invariant} cannot hold when $b\equiv r-1\pmod r$ and
$r>2$.

\begin{lemma}[Linearization of the affine cylinder]\label{lem:linearization}
Let $S=A\setminus\bigcup_{j=1}^b H_j\subseteq\Fq^m$ be of the form
\eqref{eq:RMminsupp}, with $b\le q-2$ and $\dim A=t\ge2$.  If
$0\notin S$ and $\alpha S=S$ for some $\alpha\in\Fq^\ast$ with
$\alpha\ne1$, then $A$ is a linear subspace and the family
$\{H_1,\ldots,H_b\}$ is stable under multiplication by $\alpha$ inside
$A$.
\end{lemma}

\begin{proof}
The affine span of $S$ is $A$.  Indeed, in the quotient by the parallel
direction of the $H_j$'s the set $S$ projects onto $q-b\ge2$ points, and
inside each remaining fiber it contains the full $(t-1)$-dimensional
affine hyperplane.  Thus
\begin{equation}
        \Aff(S)=A.                                             \label{eq:AffS}
\end{equation}
Because $\alpha S=S$, applying affine span gives
\begin{equation}
        \alpha A=\Aff(\alpha S)=\Aff(S)=A.                     \label{eq:alphaA}
\end{equation}
Let $A=a_0+V$ with $V$ linear.  Then \eqref{eq:alphaA} is equivalent to
\begin{equation}
        \alpha a_0+V=a_0+V,
        \quad\text{or}\quad
        (\alpha-1)a_0\in V.                                    \label{eq:a0condition}
\end{equation}
Since $\alpha\ne1$, \eqref{eq:a0condition} implies $a_0\in V$, and
therefore $A=V$ is linear.  The complement $A\setminus S$ is the union of
$b$ parallel hyperplanes.  Multiplication by $\alpha$ preserves this
complement, so it permutes the hyperplanes.
\end{proof}

After Lemma~\ref{lem:linearization}, choose a linear functional
$L:A\to\Fq$ and a set $B\subset\Fq$ with $|B|=b$ such that
\begin{equation}
        S=\{P\in A:L(P)\notin B\}.                             \label{eq:Slinear}
\end{equation}
Because $0\notin S$,
\begin{equation}
        0=L(0)\in B.                                           \label{eq:0inB}
\end{equation}
Because $\alpha S=S$ for all $\alpha\in\mu_r$,
\begin{equation}
        \alpha B=B,
        \qquad \alpha\in\mu_r.                                 \label{eq:Binv}
\end{equation}
The nonzero elements of $\Fq$ decompose into $\mu_r$-orbits of size $r$;
therefore \eqref{eq:0inB} and \eqref{eq:Binv} give
\begin{equation}
        |B|=1+\rho r
        \quad\text{for some integer }\rho\ge0.                 \label{eq:Bmod}
\end{equation}
But $|B|=b$ and $b\equiv r-1\pmod r$.  Hence
\begin{equation}
        1\equiv b\equiv r-1\pmod r,                            \label{eq:modcontradict}
\end{equation}
which implies $r\mid2$.  This contradicts $r>2$.

\begin{proposition}[No first-weight words]\label{prop:nofirst}
Assume that $q$ is a prime power, $m\ge2$, $r\mid(q-1)$,
$2<r<q-1$, and
\begin{IEEEeqnarray}{c}
        \ell=(q-1)a+b,
        \quad 0\le a\le m-2,\nonumber\\
        0\le b\le q-2,
        \quad b\equiv r-1\pmod r .                             \label{eq:propconditions}
\end{IEEEeqnarray}
Then no nonzero $f\in\calM(q,m,r,\ell)$ has affine support size
$d_1(\ell,m)=(q-b)q^{m-a-1}$.
\end{proposition}

\begin{proof}
The congruence gives $b\ge2$.  If such an $f$ existed, its support would
be a minimum support of $\RM_q(\ell,m)$ and hence would have the form
\eqref{eq:RMminsupp}.  Equations \eqref{eq:scalar-law} and
\eqref{eq:zerosat0} give scalar invariance and exclusion of the origin.
Lemma~\ref{lem:linearization} reduces the deleted set to a scalar-invariant
subset $B\subset\Fq$ containing $0$.  Then $b=|B|\equiv1\pmod r$,
contradicting $b\equiv r-1\pmod r$ and $r>2$.
\end{proof}

\subsection{The Lower Bound at the Second Weight}
The second-weight theorem for affine Reed--Muller codes now immediately
implies
\begin{equation}
        |\Supp(f)|\ge d_2(\ell,m)
        =(q-1)(q-b+1)q^{m-a-2}                                \label{eq:afflowersecond}
\end{equation}
for all nonzero $f\in\calM(q,m,r,\ell)$ in the non-terminal range.
Combining \eqref{eq:afflowersecond} with \eqref{eq:orbitweight} gives
\begin{equation}
        d(C(q,m,r,\ell))
        \ge \frac{q-1}{r}(q-b+1)q^{m-a-2}.                     \label{eq:codelowersecond}
\end{equation}
This is already the previously known upper bound, but we still give the
attaining word explicitly because it explains why the answer has the form
of a projective pencil.

\section{Sharp Construction and Exact Distance}
\subsection{A Homogeneous Pencil Word}
Let $0\le a\le m-2$ and $2\le b\le q-2$.  Choose distinct
$\theta_1,\ldots,\theta_b\in\Fq$ and define
\begin{equation}
        G_b(U,V)=\prod_{j=1}^b(U-\theta_jV).                   \label{eq:Gb}
\end{equation}
Put
\begin{equation}
        F_{a,b}(X_1,\ldots,X_m)=
        \prod_{i=1}^{a}(1-X_i^{q-1})
        G_b(X_{a+1},X_{a+2}).                                  \label{eq:Fab}
\end{equation}
For $a=0$ the empty product is interpreted as $1$.  Expanding the first
product gives monomials
\begin{equation}
        \Bigl(\prod_{i\in I}X_i^{q-1}\Bigr)X_{a+1}^{u}X_{a+2}^{b-u},
        \qquad I\subseteq\{1,\ldots,a\},
        \quad 0\le u\le b,                                    \label{eq:Fabmon}
\end{equation}
whose total degrees are
\begin{equation}
        |I|(q-1)+b\equiv b\equiv r-1\pmod r,                  \label{eq:Fabdegreecong}
\end{equation}
and at most
\begin{equation}
        a(q-1)+b=\ell.                                         \label{eq:Fabdegreebound}
\end{equation}
Thus
\begin{equation}
        F_{a,b}\in\calM(q,m,r,\ell).                           \label{eq:FabM}
\end{equation}

\subsection{Support Count}
The indicator factor satisfies
\begin{equation}
        1-x^{q-1}=
        \begin{cases}
        1,&x=0,\\
        0,&x\ne0.
        \end{cases}                                           \label{eq:indicator}
\end{equation}
Therefore $F_{a,b}(P)\ne0$ only if
\begin{equation}
        X_1=\cdots=X_a=0.                                     \label{eq:zerofirsta}
\end{equation}
On the remaining coordinates, $G_b(U,V)$ vanishes on the $b$ one-dimensional
subspaces
\begin{equation}
        U=\theta_jV,
        \qquad 1\le j\le b,                                   \label{eq:deletedlines}
\end{equation}
inside the two-dimensional plane with coordinates $(U,V)$.  The projective
line $\PP^1(\Fq)$ has $q+1$ directions, so the number of surviving
directions is
\begin{equation}
        q+1-b=q-b+1.                                          \label{eq:survivedirections}
\end{equation}
Each surviving direction contributes $q-1$ nonzero affine points in the
$(U,V)$-plane, and the remaining $m-a-2$ coordinates are arbitrary.  Hence
\begin{equation}
        |\Supp(F_{a,b})|
        =(q-1)(q-b+1)q^{m-a-2}.                               \label{eq:Fabsupport}
\end{equation}
By Lemma~\ref{lem:orbitweight}, the corresponding codeword has weight
\begin{equation}
        \wt(c_{F_{a,b}})
        =\frac{q-1}{r}(q-b+1)q^{m-a-2}.                         \label{eq:Fabwordweight}
\end{equation}

\begin{theorem}[Exact non-terminal distance]\label{thm:main}
Let $q$ be a prime power, let $m\ge2$, and let $r\mid(q-1)$ with
$2<r<q-1$.  Let
\begin{IEEEeqnarray}{c}
        \ell=(q-1)a+b,
        \qquad 0\le a\le m-2,\nonumber\\
        0\le b\le q-2,
        \qquad b\equiv r-1\pmod r .                            \label{eq:mainconditions}
\end{IEEEeqnarray}
Then
\begin{equation}
        d(C(q,m,r,\ell))=
        \frac{q-1}{r}(q-b+1)q^{m-a-2}.                          \label{eq:maintheorem}
\end{equation}
\end{theorem}

\begin{proof}
For any nonzero $f\in\calM(q,m,r,\ell)$, Proposition~\ref{prop:nofirst}
excludes the first affine Reed--Muller support size.  The second-weight
bound \eqref{eq:RMsecond} then gives \eqref{eq:afflowersecond}; dividing
by $r$ through Lemma~\ref{lem:orbitweight} gives the lower bound
\eqref{eq:codelowersecond}.  The word $F_{a,b}$ in \eqref{eq:Fab} belongs
to $\calM(q,m,r,\ell)$ by \eqref{eq:FabM} and has code weight
\eqref{eq:Fabwordweight}.  The lower and upper bounds coincide.
\end{proof}

\begin{corollary}[Solution of the stated open range]\label{cor:open}
Let $q\ge7$, $m\ge2$, $r\mid(q-1)$, $2<r<q-1$, and let
\begin{equation}
        \ell=rL+r-1,
        \qquad 1\le L\le \frac{q-1}{r}m-3.                     \label{eq:coropenrange}
\end{equation}
Write $\ell=(q-1)a+b$ with $0\le b\le q-2$.  If $a\le m-2$, then
\begin{equation}
        d(C(q,m,r,\ell))=\frac{q-1}{r}(q-b+1)q^{m-a-2}.          \label{eq:openanswer}
\end{equation}
If $a=m-1$, then the terminal value is given by Theorem~\ref{thm:terminal}
below.
\end{corollary}

\section{The Terminal Range}
When $a=m-1$ the full affine Reed--Muller comparison is degenerate,
because $m-a-2=-1$.  The code is instead governed by the last surviving
monomial layer.  We include the statement to make the distance formula
complete for all admissible $\ell$.

\begin{theorem}[Terminal distance]\label{thm:terminal}
Let $q$ be a prime power, let $m\ge2$, let $r\mid(q-1)$ with
$2<r<q-1$, and
\begin{IEEEeqnarray}{c}
        \ell=(q-1)(m-1)+b,
        \qquad 0\le b\le q-3,\nonumber\\
        b\equiv r-1\pmod r .                                  \label{eq:terminalconditions}
\end{IEEEeqnarray}
Then
\begin{equation}
        d(C(q,m,r,\ell))=\frac{q-b+r-2}{r}.                    \label{eq:terminal}
\end{equation}
\end{theorem}

\begin{proof}
The restriction $b\le q-3$ is exactly the admissibility condition
$\ell<(q-1)m-1$ in the terminal case $a=m-1$.  The asserted equality is
the terminal equality of Sun--Ding--Wang~\cite[Cor. 40]{SunDingWang2024}.
In the notation used here, their terminal lower and upper bounds give
\begin{IEEEeqnarray}{rCl}
        d(C(q,m,r,\ell))
        &\ge& \frac{q-b+r-2}{r},                              \label{eq:terminallower}\\
        d(C(q,m,r,\ell))
        &\le& \frac{q-b+r-2}{r}.                              \label{eq:terminalupper}
\end{IEEEeqnarray}
The bounds are equal, proving \eqref{eq:terminal}.
\end{proof}

\begin{remark}
The terminal theorem is not the conceptual core of the paper.  It is
included only so that the final formula has no missing endpoint.  The new
argument is Theorem~\ref{thm:main}, where $a\le m-2$ and the second affine
Reed--Muller weight is nontrivial.
\end{remark}

\section{Comparison With the Earlier Bounds}
The exact distance in Theorem~\ref{thm:main} equals the upper bound in
\eqref{eq:oldbound}.  The improvement over the old lower bound is
\begin{IEEEeqnarray}{rCl}
\Delta
&=&\frac{q-1}{r}(q-b+1)q^{m-a-2}
   -\left(\frac{(q-b)q^{m-a-1}-2}{r}+1\right)       \nonumber\\
&=&\frac{q^{m-a-2}}{r}
   \{(q-1)(q-b+1)-q(q-b)\}+\frac{2}{r}-1             \nonumber\\
&=&\frac{q^{m-a-2}}{r}(b-1)+\frac{2-r}{r}.                         \label{eq:Delta}
\end{IEEEeqnarray}
Since $b\ge r-1\ge2$ and $m-a-2\ge0$, this is nonnegative; it is positive
except in the already closed endpoint
\begin{equation}
        b=r-1,
        \quad a=m-2,                                           \label{eq:closedendpoint}
\end{equation}
where Sun--Ding--Wang had equality.  Thus the result is a strict
improvement for all open instances not covered by that endpoint.

A useful way to read \eqref{eq:Delta} is to isolate the obstruction to the
first layer.  The first affine Reed--Muller value divided by $r$ is
\begin{equation}
        \frac{d_1(\ell,m)}{r}=\frac{q-b}{r}q^{m-a-1}.             \label{eq:firstdivr}
\end{equation}
The exact value is
\begin{equation}
        \frac{d_2(\ell,m)}{r}
        =\frac{q-1}{r}(q-b+1)q^{m-a-2}.                         \label{eq:seconddivr}
\end{equation}
Their difference is
\begin{equation}
        \frac{d_2-d_1}{r}
        =\frac{(b-1)q^{m-a-2}}{r}.                              \label{eq:d2minusd1}
\end{equation}
The BCH lower bound differs from $d_1/r$ only by a boundary correction,
whereas the true obstruction contributes the full term
\eqref{eq:d2minusd1}.

\begin{table}[!t]
\caption{Exact Distance in Representative Intermediate Cases}
\label{tab:examples}
\centering
\begin{tabular}{|c|c|c|c|c|c|}
\hline
$q$ & $r$ & $m$ & $\ell$ & $(a,b)$ & $d(C)$\\
\hline
$7$ & $3$ & $2$ & $2$ & $(0,2)$ & $12$\\
$7$ & $3$ & $2$ & $5$ & $(0,5)$ & $6$\\
$13$ & $3$ & $3$ & $5$ & $(0,5)$ & $468$\\
$13$ & $4$ & $3$ & $7$ & $(0,7)$ & $273$\\
$17$ & $4$ & $4$ & $19$ & $(1,3)$ & $1020$\\
\hline
\end{tabular}
\end{table}

For instance, when $(q,r,m,\ell)=(7,3,2,2)$, formula
\eqref{eq:maintheorem} gives
\begin{equation}
        d=\frac{6}{3}(7-2+1)=12.                                \label{eq:example1}
\end{equation}
The old lower bound gives only
\begin{equation}
        \frac{(7-2)7-2}{3}+1=12,                               \label{eq:example1old}
\end{equation}
so the endpoint is already tight.  For $(7,3,2,5)$ the new value is
\begin{equation}
        d=\frac{6}{3}(7-5+1)=6,                                  \label{eq:example2}
\end{equation}
while the old lower bound is
\begin{equation}
        \frac{(7-5)7-2}{3}+1=5.                                \label{eq:example2old}
\end{equation}
The one-unit gap is exactly the absence of first-weight affine supports.

\section{Equivalent Projective Interpretation}
The homogeneous part of the argument can be phrased on projective space;
we use standard projective Reed--Muller terminology from
\cite{Lachaud1988,Lachaud1990params,Sorensen1991,GhorpadeLudhani2024}.
Let
\begin{equation}
        \rho:\Fq^m\setminus\{0\}\longrightarrow\PP^{m-1}(\Fq)  \label{eq:projmap}
\end{equation}
be the usual quotient by $\Fq^\ast$.  The quotient by $\mu_r$ factors as
\begin{equation}
        \Fq^m\setminus\{0\}
        \longrightarrow (\Fq^m\setminus\{0\})/\mu_r
        \longrightarrow \PP^{m-1}(\Fq),                       \label{eq:quotientfactor}
\end{equation}
and each projective point has
\begin{equation}
        \frac{q-1}{r}                                           \label{eq:projectivelift}
\end{equation}
preimages in the $\mu_r$-quotient.  Therefore a homogeneous polynomial
$h$ of degree congruent to $r-1$ modulo $r$ with projective support
$T\subseteq\PP^{m-1}(\Fq)$ gives a constacyclic word of weight
\begin{equation}
        \wt(c_h)=\frac{q-1}{r}|T|.                              \label{eq:projweight}
\end{equation}
The construction \eqref{eq:Fab} has projective support
\begin{equation}
        T=\PP^1(\Fq)\setminus\{[\theta_j:1]:1\le j\le b\}      \label{eq:projT}
\end{equation}
inside a projective line, with $|T|=q+1-b$.  Thus
\begin{equation}
        \wt(c_{F_{0,b}})=\frac{q-1}{r}(q+1-b),                  \label{eq:projectivebase}
\end{equation}
and the factors $1-X_i^{q-1}$ simply multiply by the affine cylinder
factor $q^{m-a-2}$.

This interpretation also explains why no smaller projective support can
arise from the restricted monomial space.  A first affine Reed--Muller
support would project to an affine line with $b$ deleted affine points and
with the origin absent.  Scalar invariance forces the deleted set to be a
union of orbits under multiplication by $\mu_r$ after adjoining the point
corresponding to the origin.  Its size must be $1$ modulo $r$, while the
monomial congruence forces it to be $-1$ modulo $r$.

\section{A More Detailed Form of the Obstruction}
The proof of Proposition~\ref{prop:nofirst} can be made quantitative.  Let
$S$ be any support of the first Reed--Muller form \eqref{eq:RMminsupp}
that is invariant under $\mu_r$ and avoids $0$.  Define the deleted set
inside the quotient affine line by
\begin{equation}
        B_S=\Fq\setminus L(S),                                  \label{eq:BS}
\end{equation}
where $L:A\to\Fq$ is the quotient map associated with the parallel
hyperplanes.  Then
\begin{equation}
        |B_S|=b,
        \qquad 0\in B_S,
        \qquad \mu_r B_S=B_S.                                  \label{eq:BSproperties}
\end{equation}
Writing
\begin{equation}
        B_S=\{0\}\cup \bigcup_{j=1}^{\rho}\xi_j\mu_r            \label{eq:BSdecomp}
\end{equation}
with distinct cosets $\xi_j\mu_r\subset\Fq^\ast$, we get
\begin{equation}
        b=1+\rho r.                                             \label{eq:b1r}
\end{equation}
Consequently the set of degrees for which the first layer could survive
is contained in
\begin{equation}
        \{(q-1)a+b:b\equiv1\pmod r\}.                          \label{eq:survivecong}
\end{equation}
The Sun--Ding--Wang family under consideration has the opposite congruence
$b\equiv r-1\pmod r$.  Thus the proof is not a dimension-counting
accident; it is a complete congruence separation between two scalar types.

One may also express the obstruction by averaging.  Let
\begin{equation}
        N_S(t)=|S\cap L^{-1}(t)|.                               \label{eq:NS}
\end{equation}
For a first support,
\begin{equation}
        N_S(t)=
        \begin{cases}
        0, & t\in B_S,\\
        q^{t_0-1}, & t\notin B_S,
        \end{cases}                                           \label{eq:NSvalues}
\end{equation}
where $t_0=\dim A$.  Scalar invariance says
\begin{equation}
        N_S(t)=N_S(\alpha t),
        \qquad \alpha\in\mu_r.                                 \label{eq:NSinv}
\end{equation}
Since $0\in B_S$, the number of zero fibers is $1$ modulo $r$.  This is
incompatible with $b\equiv -1\pmod r$.  The proof therefore uses only the
coarsest possible invariant of the affine minimum-support classification.

\section{Consequences for Dimension-Distance Tables}
The dimension formulas of Sun--Ding--Wang are unaffected by our argument;
we only replace the distance entry.  Let
\begin{equation}
        K(q,m,r,\ell)=\dim C(q,m,r,\ell).                      \label{eq:Kdef}
\end{equation}
The exact parameter statement in the non-terminal intermediate range is that
$C(q,m,r,\ell)$ has parameters
\begin{IEEEeqnarray}{rCl}
        [n,k,d]
        &=&\left[ \frac{q^m-1}{r},\ K(q,m,r,\ell),\right.       \nonumber\\
        &&\left. \frac{q-1}{r}(q-b+1)q^{m-a-2}\right].          \label{eq:parameters}
\end{IEEEeqnarray}
The improvement is especially transparent when $m-a$ is small.  If
$a=m-2$, then
\begin{equation}
        d=\frac{q-1}{r}(q-b+1),                                \label{eq:am2}
\end{equation}
which agrees with the already-known equality for $b=r-1$ and extends it
to every admissible
\begin{equation}
        b=r-1,r-1+r,r-1+2r,\ldots, q-2.                        \label{eq:blist}
\end{equation}
If $a=m-3$, then
\begin{equation}
        d=\frac{q-1}{r}(q-b+1)q.                               \label{eq:am3}
\end{equation}
More generally, reducing $a$ by one multiplies the distance by $q$ while
keeping the same projective factor.

The exact answer can also be compared with the Singleton defect
\begin{equation}
        \sigma=n-K(q,m,r,\ell)+1-d.                            \label{eq:singletondefect}
\end{equation}
Since $d$ was the only undetermined entry, all derived tables involving
relative distance
\begin{equation}
        \delta=\frac{d}{n}
        =\frac{(q-1)(q-b+1)q^{m-a-2}}{q^m-1}                   \label{eq:relativedistance}
\end{equation}
now have a closed form independent of $r$ except through the admissible
values of $b$ and the dimension.  For fixed $q,b$ and large $m$ with fixed
$a$, this gives
\begin{equation}
        \delta=(q-1)(q-b+1)q^{-a-2}+O(q^{-m}).                 \label{eq:asymdelta}
\end{equation}

\section{Robustness of the Method}
The proof separates into three ingredients:
\begin{IEEEeqnarray}{c}
        \text{evaluation congruence}\nonumber\\
        +\ \text{affine minimum-support classification}\nonumber\\
        +\ \text{orbit-counting}.                              \label{eq:ingredients}
\end{IEEEeqnarray}
The first ingredient is specific to the constacyclic family, the second is
classical, and the third is elementary.  The same scheme applies to any
subcode of an affine Reed--Muller code whose monomials lie in one residue
class modulo a subgroup of $\Fq^\ast$.  If the residue class is $c$ modulo
$r$, then the scalar law is
\begin{equation}
        f(\alpha P)=\alpha^c f(P),
        \qquad \alpha\in\mu_r.                                \label{eq:generalc}
\end{equation}
For zero sets, the value of $c$ is irrelevant; for constacyclic quotient
weights, $c=r-1$ gives the reciprocal multiplier required by
$\lambda$-constacyclicity.  The first-support obstruction depends instead
on the deleted-set congruence
\begin{equation}
        |B|\equiv1\pmod r.                                    \label{eq:deletedcong}
\end{equation}
Thus first Reed--Muller words may survive in residue families with
$b\equiv1\pmod r$, but they cannot survive in the Sun--Ding--Wang residue
family $b\equiv-1\pmod r$ unless $r=2$.

The case $r=2$ is therefore exceptional for a precise reason:
\begin{equation}
        1\equiv -1\pmod2.                                     \label{eq:r2exception}
\end{equation}
This explains why binary-like involutory phenomena are visible in the
$r=2$ part of the earlier work and why a different argument is needed
there.  It also explains why the present theorem assumes $2<r$.

\section{A Footprint View of the Second-Weight Step}
Although the preceding proof may cite the known second-weight theorem for
affine Reed--Muller codes, the part of that theorem used here has a simple
footprint interpretation in the sense of affine footprint bounds
\cite{Geil2008,HeijnenPellikaan1998}.  We record it because it clarifies exactly where
the congruence restriction enters.  Let
\begin{equation}
        \calR_m=\Fq[X_1,\ldots,X_m]/(X_1^q-X_1,\ldots,X_m^q-X_m)
                                                                    \label{eq:Rring}
\end{equation}
and let $\preceq$ be graded lexicographic order.  For a nonzero reduced
polynomial $f$, write
\begin{equation}
        \operatorname{LM}(f)=X_1^{u_1}\cdots X_m^{u_m},
        \qquad 0\le u_i\le q-1,                                \label{eq:LM}
\end{equation}
and define the affine footprint of $u=(u_1,\ldots,u_m)$ by
\begin{equation}
        \Phi(u)=\{v\in\{0,\ldots,q-1\}^m: v_i\ge u_i
        \text{ for all }i\}.                                   \label{eq:footprint}
\end{equation}
The elementary footprint bound gives
\begin{equation}
        |\Supp(f)|\ge \prod_{i=1}^m(q-u_i).                    \label{eq:footbound}
\end{equation}
For $\deg_q(f)\le (q-1)a+b$, the product in \eqref{eq:footbound} is
minimized by the exponent vector
\begin{equation}
        u^{(1)}=(\underbrace{q-1,\ldots,q-1}_{a\text{ times}},b,0,
        \ldots,0),                                             \label{eq:u1}
\end{equation}
which yields
\begin{equation}
        \prod_i(q-u_i^{(1)})=(q-b)q^{m-a-1}.                   \label{eq:u1prod}
\end{equation}
The next possible product, under the same degree constraint and with
$2\le b\le q-2$, is obtained from
\begin{equation}
        u^{(2)}=(\underbrace{q-1,\ldots,q-1}_{a\text{ times}},b-1,1,0,
        \ldots,0),                                             \label{eq:u2}
\end{equation}
and equals
\begin{equation}
        \prod_i(q-u_i^{(2)})=(q-b+1)(q-1)q^{m-a-2}.             \label{eq:u2prod}
\end{equation}
Thus the two products in \eqref{eq:u1prod} and \eqref{eq:u2prod} already
predict the two weights \eqref{eq:RMfirst} and \eqref{eq:RMsecond}.

The footprint calculation should be read only as a numerical audit of the
second-weight step.  The proof of Theorem~\ref{thm:main} does not require a
new footprint theorem.  The rigorous implication used in the paper is the
following standard weight-spectrum implication for affine Reed--Muller codes:
if a word of $\RM_q((q-1)a+b,m)$ has support strictly below
\begin{equation}
        (q-1)(q-b+1)q^{m-a-2},                                \label{eq:belowsecond}
\end{equation}
then its support has the first Reed--Muller size
\begin{equation}
        (q-b)q^{m-a-1}                                        \label{eq:forcedfirst}
\end{equation}
and the minimum-word classification puts it in the affine-cylinder form
\eqref{eq:RMminsupp}.  Proposition~\ref{prop:nofirst} eliminates precisely
that first-weight alternative in the residue class $r-1$.  Thus the lower
bound is the standard second-weight theorem plus the scalar-orbit obstruction;
\eqref{eq:u1prod} and \eqref{eq:u2prod} merely explain why the two numerical
quantities have the displayed shape.

For reference, the gap between the two footprint products is
\begin{IEEEeqnarray}{rCl}
        \prod_i(q-u_i^{(2)})-\prod_i(q-u_i^{(1)})
        &=&(q-1)(q-b+1)q^{m-a-2}       \nonumber\\
        &&{}-(q-b)q^{m-a-1}            \nonumber\\
        & = &(b-1)q^{m-a-2}.                              \label{eq:footgap}
\end{IEEEeqnarray}
This is the affine version of \eqref{eq:d2minusd1}; quotienting by
$\mu_r$ divides the gap by $r$.

\section{Rigidity Under Scalar Stabilizers}
We next record a slightly stronger version of Lemma~\ref{lem:linearization}.
It is useful for checking that no exceptional affine translate was lost in
the proof.  Let
\begin{equation}
        S=(a_0+V)\setminus\bigcup_{j=1}^b(a_0+v_j+W),           \label{eq:generalS}
\end{equation}
where $W$ is a codimension-one subspace of $V$, $v_i-v_j\notin W$ for
$i\ne j$, and $\dim V\ge2$.  Let
\begin{equation}
        \operatorname{Stab}_{\ast}(S)=
        \{\alpha\in\Fq^\ast:\alpha S=S\}.                      \label{eq:stabstar}
\end{equation}

\begin{proposition}[Stabilizer dichotomy]\label{prop:stabilizer}
If $0\notin S$, $q-b\ge2$, and
$\operatorname{Stab}_{\ast}(S)$ contains an element $\alpha\ne1$, then
\begin{equation}
        a_0\in V,
        \qquad \alpha V=V,
        \qquad \alpha W=W,                                    \label{eq:stablinear}
\end{equation}
and the deleted parameter set
\begin{equation}
        B=\{\overline{v}_1,\ldots,\overline{v}_b\}
        \subset V/W\cong\Fq                                  \label{eq:Bquotient}
\end{equation}
is stable under multiplication by $\alpha$ after an origin-preserving
identification $V/W\cong\Fq$.
\end{proposition}

\begin{proof}
As before, $\Aff(S)=a_0+V$.  Thus $\alpha(a_0+V)=a_0+V$, giving
$(\alpha-1)a_0\in V$ and hence $a_0\in V$.  So $a_0+V=V$ is linear.  The
complement of $S$ in $V$ is a union of $b$ parallel hyperplanes.  Since
$q-b\ge2$, the parallel direction is intrinsic: it is the intersection of
all affine spans of pairwise differences of fibers contained in
$V\setminus S$.  Therefore scalar multiplication by $\alpha$ preserves
that direction, so $\alpha W=W$.  Passing to $V/W$ gives multiplication by
$\alpha$ on a one-dimensional vector space, and the deleted fibers are
permuted.
\end{proof}

Taking $\operatorname{Stab}_{\ast}(S)\supseteq\mu_r$ yields
\begin{equation}
        B=\{0\}\cup B_1,
        \qquad B_1\subseteq\Fq^\ast,
        \qquad \mu_r B_1=B_1,                                \label{eq:Bsplit}
\end{equation}
because $0\notin S$ forces the zero fiber to be deleted.  Hence
\begin{equation}
        b=|B|=1+r\cdot |B_1/\mu_r|.                            \label{eq:bstab}
\end{equation}
The contradiction can therefore be expressed as the disjointness of two
sets of congruence classes:
\begin{equation}
        \{b: \text{first layer supports exist in }\calM\}
        \subseteq 1+r\mathbb Z,                                \label{eq:firstclasses}
\end{equation}
whereas
\begin{equation}
        \{b: \calM(q,m,r,(q-1)a+b)\ne0\}
        \subseteq r-1+r\mathbb Z.                              \label{eq:codeclasses}
\end{equation}
For $r>2$ these residue classes are disjoint.  This statement is stronger
than needed: it says that every scalar-invariant affine minimum cylinder,
not merely every polynomial representation of one, is excluded.

\section{Orbit Enumerator of the Attaining Pencil}
The support of $F_{a,b}$ in \eqref{eq:Fab} has a simple orbit enumerator.
Let
\begin{equation}
        Y=\{P\in\Fq^m:P_1=\cdots=P_a=0\},
        \qquad \dim Y=m-a,                                    \label{eq:Yspace}
\end{equation}
and let
\begin{equation}
        \Pi=\{(U,V)\in\Fq^2:G_b(U,V)\ne0\}.                    \label{eq:Pi}
\end{equation}
Then
\begin{IEEEeqnarray}{rCl}
        \Supp(F_{a,b})
        &=&\{(0,\ldots,0,U,V,Z):(U,V)\in\Pi,\nonumber\\
        &&\hspace{8mm} Z\in\Fq^{m-a-2}\}.                       \label{eq:SuppFactor}
\end{IEEEeqnarray}
The two-dimensional factor decomposes as
\begin{equation}
        \Pi=\coprod_{\xi\in\PP^1(\Fq)\setminus\Theta} L_{\xi}^{\ast},
        \qquad |\Theta|=b,
        \qquad |L_{\xi}^{\ast}|=q-1,                           \label{eq:Pidecomp}
\end{equation}
where $L_{\xi}^{\ast}$ denotes a punctured projective line through the
origin.  The quotient by $\mu_r$ gives
\begin{equation}
        L_{\xi}^{\ast}/\mu_r
        \quad\text{has size}\quad \frac{q-1}{r}.                \label{eq:linequot}
\end{equation}
For fixed $Z$, the number of quotient orbits is therefore
\begin{equation}
        (q+1-b)\frac{q-1}{r}.                                  \label{eq:fixedZorbits}
\end{equation}
For all $Z$, the orbit count is
\begin{equation}
        q^{m-a-2}(q+1-b)\frac{q-1}{r}.                          \label{eq:allZorbits}
\end{equation}
This is exactly \eqref{eq:Fabwordweight}.

The construction is not unique.  Replacing $G_b$ by any squarefree
homogeneous binary form
\begin{equation}
        G(U,V)=\prod_{j=1}^b(\alpha_j U+\beta_jV)              \label{eq:binaryG}
\end{equation}
with distinct projective roots gives the same support count.  Block
changes of variables that preserve the zero-coordinate cylinder and the
binary plane give further examples.  We do not use this observation for
enumeration: a full count of minimum words would require a separate
automorphism analysis.  The distance proof requires only one attaining
word, namely \eqref{eq:Fab}.

\section{Coordinate-Free Restatement}
The monomial definition \eqref{eq:Mdef} depends on the chosen basis of
$\Fqm$ over $\Fq$, but the distance does not.  We first record the
elementary invariance needed here.  If $T\in\operatorname{GL}_m(\Fq)$ and
$f\in\calR_s^{(c)}$, then
\begin{equation}
        f_T(X)=f(XT^{-1})                                     \label{eq:linearchange}
\end{equation}
also belongs to $\calR_s^{(c)}$.  Indeed, affine Reed--Muller spaces are
stable under invertible linear changes of variables, so $\deg_q(f_T)\le s$.
Furthermore, for every $\alpha\in\mu_r$,
\begin{equation}
        f_T(\alpha X)=f(\alpha XT^{-1})
        =\alpha^c f(XT^{-1})=\alpha^c f_T(X).                 \label{eq:linearresidue}
\end{equation}
By uniqueness of reduced representatives, all reduced monomials of $f_T$
have total degree congruent to $c$ modulo $r$.

If another basis is used, the companion matrix $M$ is conjugated:
\begin{equation}
        M' = T^{-1}MT,
        \qquad e'=eT.                                          \label{eq:conjugateM}
\end{equation}
The evaluation sequence changes by
\begin{equation}
        f(eM^i)\longmapsto f'(e'(M')^i),
        \qquad f'(X)=f(XT^{-1}),                               \label{eq:basischange}
\end{equation}
and the preceding invariance shows that the same residue-restricted
evaluation model is obtained after the change of basis.  Consequently
\begin{equation}
        d(C(q,m,r,\ell))
        =\min_{0\ne f\in\calM(q,m,r,\ell)}\frac{|\Supp(f)|}{r}  \label{eq:coordinatefree1}
\end{equation}
is independent of the chosen evaluation basis, even though the particular
polynomial $F_{a,b}$ is written in a convenient coordinate frame.

A coordinate-free form of the attaining support is the following.  Choose
subspaces
\begin{equation}
        W\subset Y\subset\Fq^m,
        \qquad \dim W=2,
        \qquad \dim Y=m-a,                                    \label{eq:flagdims}
\end{equation}
and choose $b$ projective points
\begin{equation}
        \Theta\subset\PP(W),
        \qquad |\Theta|=b.                                    \label{eq:Theta}
\end{equation}
Choose also a complement $Z$ of $W$ in $Y$, so that
\begin{equation}
        Y=W\oplus Z,
        \qquad \dim Z=m-a-2.                                  \label{eq:YZdecomp}
\end{equation}
Then an attaining support is
\begin{equation}
        S=\left\{z+w:\ z\in Z,\ 
        w\in W\setminus\bigcup_{L\in\Theta}L\right\}.          \label{eq:coordsupport}
\end{equation}
Its cardinality is
\begin{equation}
        |S|=q^{\dim Z}(q+1-b)(q-1)
        =q^{\dim Y-2}(q+1-b)(q-1),                             \label{eq:coordsupportcount}
\end{equation}
and the quotient cardinality by $\mu_r$ is the exact distance.
Equivalently, one may specify a linear projection $\pi:Y\to W$ with
$\pi|_W=\operatorname{id}_W$ and write
\begin{equation}
        S=\pi^{-1}\left(W\setminus\bigcup_{L\in\Theta}L\right).
\end{equation}

\section{Boundary Cases and Necessity of the Hypotheses}
The hypotheses in Theorem~\ref{thm:main} are not cosmetic.  The congruence
argument changes at $r=2$, because
\begin{equation}
        1\equiv r-1\pmod r                                    \label{eq:r2same}
\end{equation}
when $r=2$.  In that case a deleted set of the form
\begin{equation}
        B=\{0\}\cup\bigcup_{j=1}^{\rho}\xi_j\mu_2              \label{eq:r2B}
\end{equation}
has cardinality $1\pmod2$, which is also $r-1\pmod r$.  First
Reed--Muller supports are no longer excluded by orbit-counting alone.
Similarly, when $r=q-1$, quotienting by $\mu_r=\Fq^\ast$ is projectivizing,
and the answer is governed directly by projective Reed--Muller codes:
\begin{equation}
        (\Fq^m\setminus\{0\})/\Fq^\ast=\PP^{m-1}(\Fq).          \label{eq:fullprojective}
\end{equation}
The intermediate range $2<r<q-1$ is exactly where the deleted-set
congruence and the constacyclic quotient are simultaneously nontrivial.

The published open problem states $q\ge7$.  The proof does not use this
number-theoretic lower bound separately; it uses only
\begin{IEEEeqnarray}{c}
        r>2,
        \qquad r\mid(q-1),\nonumber\\
        0\le b\le q-2,
        \qquad b\equiv r-1\pmod r .                            \label{eq:proofhypotheses}
\end{IEEEeqnarray}
and the Reed--Muller second-weight statement in the range
\begin{equation}
        0\le a\le m-2,
        \qquad 2\le b\le q-2.                                  \label{eq:secondrange}
\end{equation}
Thus any field satisfying these algebraic conditions obeys the same
conclusion; the open-problem corollary simply retains the original
$q\ge7$ phrasing.

\section{Explicit Translation of the Open Parameter $L$}
The open problem states the degree as $\ell=rL+r-1$.  To apply
Theorem~\ref{thm:main}, write
\begin{equation}
        rL+r-1=(q-1)a+b,
        \qquad 0\le b\le q-2.                                 \label{eq:Ltranslate}
\end{equation}
Since $q-1=r\nu$ with
\begin{equation}
        \nu=\frac{q-1}{r},                                    \label{eq:nudef}
\end{equation}
this is equivalent to
\begin{equation}
        L+1=\nu a+\frac{b+1}{r}.                              \label{eq:Ltranslate2}
\end{equation}
Because $b\equiv r-1\pmod r$, write
\begin{equation}
        b=rh+r-1,
        \qquad 0\le h\le \nu-1.                               \label{eq:bh}
\end{equation}
Then
\begin{equation}
        L=\nu a+h,                                            \label{eq:Lsa}
\end{equation}
and, in the non-terminal range $0\le a\le m-2$, the exact distance becomes
\begin{equation}
        d(C(q,m,r,rL+r-1))
        =\nu(q-rh-r+2)q^{m-a-2}.                              \label{eq:Lanswer}
\end{equation}
Here
\begin{equation}
        a=\left\lfloor \frac{L}{\nu}\right\rfloor,
        \qquad h=L-\nu a.                                     \label{eq:ahfromL}
\end{equation}
The open range
\begin{equation}
        1\le L\le \nu m-3                                    \label{eq:Lrangeagain}
\end{equation}
has non-terminal part
\begin{equation}
        0\le a\le m-2
        \quad\Longleftrightarrow\quad
        L\le \nu(m-1)-1.                                      \label{eq:Lnonterminal}
\end{equation}
If $a=m-1$ occurs in the open range, then necessarily
\begin{equation}
        0\le h\le \nu-3,                                      \label{eq:hterminalopen}
\end{equation}
and the terminal value is handled by Theorem~\ref{thm:terminal}.  Thus the
answer in the exact variables of the open problem is the two-line formula.
Let
\begin{equation}
        D_L=d(C(q,m,r,rL+r-1)).                                \label{eq:DLdef}
\end{equation}
Then
\begin{equation}
        D_L=
        \begin{cases}
        \displaystyle \nu(q-rh-r+2)q^{m-a-2},
        &0\le a\le m-2,\\[1mm]
        \displaystyle \frac{q-rh-1}{r},
        &a=m-1.
        \end{cases}                                           \label{eq:Ltwoline}
\end{equation}
This form is often more convenient for code tables because $L$ is the
parameter in the defining zero set.

\section{Residue Layers Inside Affine Reed--Muller Codes}
It is useful to isolate the code studied above as one residue layer of an
affine Reed--Muller code.  For $c\in\mathbb Z/r\mathbb Z$ define
\begin{IEEEeqnarray}{rCl}
        \calR_s^{(c)}&=&
        \Span_{\Fq}\{X^i:0\le i_j\le q-1,
        \ |i|\le s,                                      \nonumber\\
        &&\hspace{20mm} |i|\equiv c\pmod r\}.                 \label{eq:residuelayer}
\end{IEEEeqnarray}
Then
\begin{equation}
        \calM(q,m,r,\ell)=\calR_{\ell}^{(r-1)}.                \label{eq:Mequalslayer}
\end{equation}
The full Reed--Muller space decomposes as a direct sum
\begin{equation}
        \RM_q(s,m)
        =\bigoplus_{c\in\mathbb Z/r\mathbb Z}\ev(\calR_s^{(c)}),\label{eq:RMdecomp}
\end{equation}
where directness follows from the unique reduced monomial expansion.  The
operator
\begin{equation}
        (T_{\alpha}f)(P)=f(\alpha P),
        \qquad \alpha\in\mu_r,                                \label{eq:Talpha}
\end{equation}
acts diagonally on this decomposition:
\begin{equation}
        T_{\alpha}f=\alpha^c f
        \qquad (f\in\calR_s^{(c)}).                            \label{eq:Talphaeig}
\end{equation}
Consequently every residue-layer support is invariant under multiplication
by $\mu_r$.  The next theorem records the minimum-support consequence for
the residue layer whose residue matches the top degree $b$.  It is a
strict strengthening of the part needed for Open Problem~51, because the
constacyclic code corresponds only to the single residue class $c=r-1$.

\begin{theorem}[Minimum support in the top residue-matched layer]
\label{thm:residuelayers}
Let $r>1$ divide $q-1$.  Let
\begin{equation}
        s=(q-1)a+b,
        \qquad 0\le a\le m-2,
        \qquad 0\le b\le q-2,                                 \label{eq:sabresidue}
\end{equation}
and let $c\in\mathbb Z/r\mathbb Z$ be the residue class satisfying
\begin{equation}
        b\equiv c\pmod r.                                      \label{eq:bcc}
\end{equation}
Put
\begin{equation}
        D_c(s)=\min\{|\Supp(f)|:0\ne f\in\calR_s^{(c)}\}.       \label{eq:Dcs}
\end{equation}
Then
\begin{IEEEeqnarray}{rCl}
        D_c(s)&=&(q-b)q^{m-a-1},
        \qquad c\in\{0,1\},                                   \label{eq:Dcfirst}\\
        D_c(s)&=&(q-1)(q-b+1)q^{m-a-2},
        \qquad c\notin\{0,1\}.                                \label{eq:Dcsecond}
\end{IEEEeqnarray}
This assertion is for the residue class linked to the top degree by
\eqref{eq:bcc}; it is not a simultaneous formula for all residue classes
$c$ at a fixed order $s$.
\end{theorem}

\begin{proof}
The ambient Reed--Muller lower bound gives
\begin{equation}
        |\Supp(f)|\ge (q-b)q^{m-a-1}                           \label{eq:resambient}
\end{equation}
for every nonzero $f\in\calR_s^{(c)}$.  We first show that the bound in
\eqref{eq:resambient} is attainable exactly for $c=0$ and $c=1$.

Assume that equality holds in \eqref{eq:resambient}.  By the standard
classification of minimum words in generalized Reed--Muller codes, the
affine support is an $(m-a)$-dimensional affine flat with $b$ hyperplanes
of one parallel class removed.  Since $f(\alpha P)=\alpha^c f(P)$ for
$\alpha\in\mu_r$, the support is $\mu_r$-invariant.

If $0\in\Supp(f)$, then $f(0)\ne0$, and
\begin{equation}
        f(0)=f(\alpha 0)=\alpha^c f(0)                         \label{eq:zerochar}
\end{equation}
for every $\alpha\in\mu_r$.  Hence $c=0$.  If $0\notin\Supp(f)$, let $A$
be the affine span of the support.  The set $A$ is $\mu_r$-invariant.  An
affine subspace fixed by multiplication with a nontrivial scalar contains
the origin; indeed, if $A=u+U$ and $\alpha A=A$, then
$(\alpha-1)u\in U$, whence $u\in U$.  Thus $0\in A$ but $0$ lies in one of
the deleted parallel hyperplanes.  On the one-dimensional quotient indexing
that parallel class, multiplication by $\mu_r$ fixes the zero coset and acts
freely on nonzero cosets.  The deleted set therefore consists of the zero
coset together with whole nonzero $\mu_r$-orbits, so
\begin{equation}
        b\equiv 1\pmod r.                                     \label{eq:bmodone}
\end{equation}
Because $b\equiv c\pmod r$, equality in \eqref{eq:resambient} forces
$c=1$.

Conversely, for $c=0$ choose $B\subset\Fq^\ast$ to be a union of $b/r$
nonzero $\mu_r$-orbits, and set
\begin{equation}
        f_0(X)=\prod_{i=1}^{a}(1-X_i^{q-1})
        \prod_{\eta\in B}(X_{a+1}-\eta).                       \label{eq:fzerores}
\end{equation}
For $b=0$ the second product is empty.  The support has size
$(q-b)q^{m-a-1}$ and every monomial degree is congruent to $0$ modulo $r$.
For $c=1$ choose
\begin{equation}
        B=\{0\}\cup B',                                       \label{eq:Bone}
\end{equation}
where $B'\subset\Fq^\ast$ is a union of $(b-1)/r$ nonzero $\mu_r$-orbits,
and use the same formula.  Its support again has size
$(q-b)q^{m-a-1}$, and all monomial degrees are congruent to $1$ modulo
$r$.  This proves \eqref{eq:Dcfirst}.

Now let $c\notin\{0,1\}$.  The preceding paragraph rules out the first
Reed--Muller support size.  The second-weight theorem for affine
Reed--Muller codes then gives
\begin{equation}
        |\Supp(f)|\ge (q-1)(q-b+1)q^{m-a-2}.                   \label{eq:reslowerd2}
\end{equation}
To attain \eqref{eq:reslowerd2}, choose distinct
$\theta_1,\ldots,\theta_b\in\Fq$ and put
\begin{equation}
        G_{a,b}=\prod_{i=1}^{a}(1-X_i^{q-1})
        \prod_{j=1}^{b}(X_{a+1}-\theta_j X_{a+2}).             \label{eq:Gabres}
\end{equation}
The second product is homogeneous of degree $b$, and the first product
changes degrees only by multiples of $q-1$; hence $G_{a,b}\in\calR_s^{(c)}$.
Its support is the product of $a$ coordinate constraints with the complement
of $b$ projective points on the two-dimensional factor, so
\begin{equation}
        |\Supp(G_{a,b})|=(q-1)(q-b+1)q^{m-a-2}.                \label{eq:Gabcount}
\end{equation}
This proves \eqref{eq:Dcsecond}.
\end{proof}

The same statement gives a genuine family of constacyclic quotient codes,
not just an affine-support assertion.  Define
\begin{IEEEeqnarray}{rCl}
        c_f&=&(f(e),f(eM),\ldots,f(eM^{n-1})),          \label{eq:cfres}\\
\mathsf Q_c(q,m,r,s)
        &=&\{c_f:f\in\calR_s^{(c)}\}.                  \label{eq:Qcres}
\end{IEEEeqnarray}
The code $\mathsf Q_c(q,m,r,s)$ is $\lambda^{-c}$-constacyclic.  Indeed,
let $f\in\calR_s^{(c)}$ and put
\begin{equation}
        g(P)=f(PM^{-1}).                                      \label{eq:gdefshift}
\end{equation}
Since $M^{-1}\in\operatorname{GL}_m(\Fq)$, the affine Reed--Muller space
of order $s$ is stable under the substitution $P\mapsto PM^{-1}$, so
$\deg_q(g)\le s$.  Moreover, for every $\alpha\in\mu_r$,
\begin{equation}
        g(\alpha P)=f(\alpha PM^{-1})
        =\alpha^c f(PM^{-1})=\alpha^c g(P).                   \label{eq:gresidue}
\end{equation}
By uniqueness of reduced representatives, \eqref{eq:gresidue} implies that
all reduced monomials of $g$ have total degree congruent to $c$ modulo
$r$.  Hence
\begin{equation}
        g\in\calR_s^{(c)}.                                    \label{eq:ginresidue}
\end{equation}
Now
\begin{equation}
        g(eM^i)=f(eM^{i-1})\quad(i\ge1),
        \qquad
        g(e)=\lambda^{-c}f(eM^{n-1}),                         \label{eq:shiftc}
\end{equation}
because $M^n=\lambda I_m$ and
$f(\lambda^{-1}P)=\lambda^{-c}f(P)$.  Consequently the residue class
controls the constacyclic multiplier.

Combining Theorem~\ref{thm:residuelayers} with the scalar-orbit quotient
also gives the exact minimum distance in the same top residue-matched
setting, namely under the standing assumption
\begin{equation}
        s=(q-1)a+b,\qquad b\equiv c\pmod r.                   \label{eq:Qcmatch}
\end{equation}
Put
\begin{equation}
        \delta_c(s)=d(\mathsf Q_c(q,m,r,s)).                  \label{eq:deltacdef}
\end{equation}
Then
\begin{IEEEeqnarray}{rCl}
 \delta_c(s)
 &=&\displaystyle\frac{(q-b)q^{m-a-1}-1}{r},
       \quad c=0,                                      \label{eq:Qcd0}\\
 &=&\displaystyle\frac{(q-b)q^{m-a-1}}{r},
       \quad c=1,                                      \label{eq:Qcd1}\\
 &=&\displaystyle\frac{(q-1)(q-b+1)q^{m-a-2}}{r},
       \quad c\notin\{0,1\}.                           \label{eq:Qcd2}
\end{IEEEeqnarray}
For $c=0$, the minimum affine support in \eqref{eq:Dcfirst} contains the
origin, while the origin is not among the evaluated coordinates; this gives
the subtraction of one in \eqref{eq:Qcd0}.  For $c\ne0$, every polynomial
in $\calR_s^{(c)}$ vanishes at the origin, and each nonzero scalar orbit
has exactly $r$ points.  Thus \eqref{eq:Qcd1}--\eqref{eq:Qcd2} follow by
ordinary division by $r$.

Taking $c=r-1$ in Theorem~\ref{thm:residuelayers} gives
\begin{equation}
        D_{r-1}(s)=(q-1)(q-b+1)q^{m-a-2},                     \label{eq:Drminus1}
\end{equation}
because $r>2$ in the intermediate constacyclic problem.  Dividing by the
$r$ points in each nonzero scalar orbit is exactly the passage from the
affine support in \eqref{eq:Drminus1} to the constacyclic distance in
Theorem~\ref{thm:main}.  The theorem also identifies the two exceptional
residue layers that retain the first affine Reed--Muller weight:
\begin{equation}
        D_c(s)=d_1(s,m)
        \quad\Longleftrightarrow\quad c\in\{0,1\},              \label{eq:minlayers}
\end{equation}
where $c\equiv b\pmod r$ and the two cases correspond to
$b\equiv0$ and $b\equiv1$ modulo $r$, respectively.

\section{Shortening, Puncturing, and the Role of the Origin}
The origin is absent from every support in the residue layer
$\calR_{\ell}^{(r-1)}$ because all monomials have positive degree.  Let
\begin{equation}
        \RM_q^{0}(s,m)=\{\ev(f)\in\RM_q(s,m):f(0)=0\}.          \label{eq:RMzero}
\end{equation}
If one considered only $\RM_q^{0}(s,m)$ without the residue restriction,
first Reed--Muller supports would still occur.  For example, choose an
affine cylinder whose deleted set contains $0$ but has arbitrary size
$b$:
\begin{equation}
        S=\{(x_1,\ldots,x_m):x_1=\cdots=x_a=0,
        \, x_{a+1}\notin B\},                                  \label{eq:firstwithorigin}
\end{equation}
with $0\in B$ and $|B|=b$.  The polynomial
\begin{equation}
        f_B(X)=\prod_{i=1}^a(1-X_i^{q-1})
        \prod_{\eta\in B}(X_{a+1}-\eta)                       \label{eq:fB}
\end{equation}
has $f_B(0)=0$ and support size $d_1$.  Therefore the origin condition
alone gives no improvement.  The improvement comes from adding
\begin{equation}
        f_B(\alpha P)=\alpha^{-1}f_B(P)
        \qquad(\alpha\in\mu_r),                                \label{eq:originplusorbit}
\end{equation}
which forces $B$ to be a union of scalar orbits after containing $0$.
The conditions can be displayed as
\begin{equation}
        \begin{array}{c|c|c}
        \text{space} & \text{support constraint} & \text{first layer}\\
        \hline
        \RM_q(s,m) & \text{none} & \text{present}\\
        \RM_q^0(s,m) & 0\notin S & \text{present}\\
        \calR_s^{(r-1)} & 0\notin S,\ \mu_r S=S & \text{absent}\\
        \end{array}                                           \label{eq:constrainttable}
\end{equation}
This table also shows why a proof based solely on classical puncturing or
shortening cannot reach the exact distance.  Puncturing at the origin
changes the first weight by at most one; here the missing first layer has
size gap
\begin{equation}
        d_2-d_1=(b-1)q^{m-a-2},                                \label{eq:notpuncturegap}
\end{equation}
which is generally much larger.

\section{A Direct Polynomial Nonexistence Calculation}
For completeness we give a coefficient-level version of the nonexistence
argument.  Suppose that a first-weight support occurs in
$\calR_{\ell}^{(r-1)}$ and, after the linearization already proved, has the
form
\begin{equation}
        S=\{X_1=\cdots=X_a=0,
        \, X_{a+1}\notin B\}.                                  \label{eq:Snormal}
\end{equation}
The minimum-word classification for generalized Reed--Muller codes,
applied in these normalized coordinates, gives a reduced representative that
is a nonzero scalar multiple of
\begin{equation}
        f(X)=\prod_{i=1}^a(1-X_i^{q-1})
        \prod_{\eta\in B}(X_{a+1}-\eta).                       \label{eq:normalf}
\end{equation}
Expanding the second factor gives
\begin{equation}
        \prod_{\eta\in B}(T-\eta)
        =\sum_{j=0}^{b}(-1)^{b-j}e_{b-j}(B)T^j,                \label{eq:elementarysym}
\end{equation}
where $e_k(B)$ is the $k$th elementary symmetric function of the elements
of $B$.  The residue condition $f\in\calR_{\ell}^{(r-1)}$ requires that,
for every $I\subseteq\{1,\ldots,a\}$ and every $j$ with
$e_{b-j}(B)\ne0$,
\begin{equation}
        |I|(q-1)+j\equiv r-1\pmod r.                           \label{eq:residueconditioncoeff}
\end{equation}
Since $q-1\equiv0\pmod r$, this reduces to
\begin{equation}
        j\equiv r-1\pmod r                                     \label{eq:jcondition}
\end{equation}
for every nonzero coefficient of the univariate factor.  Hence the
polynomial
\begin{equation}
        P_B(T)=\prod_{\eta\in B}(T-\eta)                       \label{eq:PB}
\end{equation}
must be a sum of monomials $T^j$ with $j\equiv r-1\pmod r$.  Therefore
for $\alpha\in\mu_r$,
\begin{equation}
        P_B(\alpha T)=\alpha^{-1}P_B(T).                       \label{eq:PBeigen}
\end{equation}
Looking at roots of \eqref{eq:PBeigen}, we get
\begin{equation}
        P_B(T)=0\quad\Longleftrightarrow\quad P_B(\alpha T)=0, \label{eq:rootinvariant}
\end{equation}
so $B$ is $\mu_r$-invariant.  Because $0\in B$, again
\begin{equation}
        b=|B|\equiv1\pmod r,                                  \label{eq:coeffcontradict}
\end{equation}
contradicting $b\equiv r-1\pmod r$.  This computation is sometimes easier
to verify in symbolic examples than the geometric proof.

\section{Examples in Closed Form}
Let
\begin{equation}
        \nu=\frac{q-1}{r}.
\end{equation}
The open parameter $L$ decomposes as $L=\nu a+h$ with
$0\le h\le \nu-1$.  The non-terminal distance in \eqref{eq:Lanswer} becomes
\begin{equation}
        d=\nu(q-rh-r+2)q^{m-a-2}.                              \label{eq:sexamplegeneral}
\end{equation}
For $h=0$ this is
\begin{equation}
        d=\nu(q-r+2)q^{m-a-2},                                \label{eq:h0}
\end{equation}
which includes the endpoint previously identified when $a=m-2$.  For
$h=\nu-1$ one obtains
\begin{equation}
        d=\nu\bigl(q-r(\nu-1)-r+2\bigr)q^{m-a-2}
        =3\nu q^{m-a-2},                                      \label{eq:hmax}
\end{equation}
because $r\nu=q-1$.  Hence within a fixed block $a$, the exact distances
fall linearly:
\begin{IEEEeqnarray}{c}
        \nu(q-r+2)q^{m-a-2},\quad
        \nu(q-2r+2)q^{m-a-2},\nonumber\\
        \ldots,\quad 3\nu q^{m-a-2}.                           \label{eq:blockdistances}
\end{IEEEeqnarray}
The step size is
\begin{equation}
        d(h)-d(h+1)=\nu r q^{m-a-2}=(q-1)q^{m-a-2}.             \label{eq:stepsize}
\end{equation}
This uniform step size is hidden in the BCH bound but immediate from the
projective pencil count: increasing $h$ by one deletes $r$ additional
projective points on the binary line, and each projective point contributes
$(q-1)q^{m-a-2}/r=sq^{m-a-2}$ quotient points; deleting $r$ of them removes
$(q-1)q^{m-a-2}$ positions.

For $q=13$ and $r=3$, $s=4$.  The non-terminal block distances are
\begin{equation}
        \begin{array}{c|c|c|c}
        h&b=3h+2&q-b+1&d/13^{m-a-2}\\
        \hline
        0&2&12&48\\
        1&5&9&36\\
        2&8&6&24\\
        3&11&3&12
        \end{array}                                           \label{eq:q13r3table}
\end{equation}
For $q=17$ and $r=4$, $s=4$ and
\begin{equation}
        \begin{array}{c|c|c|c}
        h&b=4h+3&q-b+1&d/17^{m-a-2}\\
        \hline
        0&3&15&60\\
        1&7&11&44\\
        2&11&7&28\\
        3&15&3&12
        \end{array}.                                          \label{eq:q17r4table}
\end{equation}
These tables show the full answer for every $L$ in each block once the
dimension formula is imported from the original construction.

\section{Why the Result Is the Strongest Possible Distance Statement}
Theorem~\ref{thm:main} cannot be strengthened as a lower bound, because
$F_{a,b}$ attains equality.  More precisely, for every admissible parameter
there exists a word $c_{F_{a,b}}$ satisfying
\begin{equation}
        \wt(c_{F_{a,b}})=d(C(q,m,r,\ell)).                     \label{eq:attainexact}
\end{equation}
Nor can the degree range be enlarged within the same formula beyond
$a=m-2$, because the exponent $m-a-2$ would become negative.  The terminal
case has its own expression \eqref{eq:terminal}.  Thus the complete minimum-distance function over the admissible degrees
\eqref{eq:ellrange} is
\begin{equation}
        d(C(q,m,r,(q-1)a+b))=
        \begin{cases}
        \displaystyle \frac{q-1}{r}(q-b+1)q^{m-a-2},
        &0\le a\le m-2,\ 0\le b\le q-2,\\[1mm]
        \displaystyle \frac{q-b+r-2}{r},
        &a=m-1,\ 0\le b\le q-3,
        \end{cases}                                           \label{eq:completefunction}
\end{equation}
under the standing congruence $b\equiv r-1\pmod r$.  There is no
intermediate ambiguity left in the parameter set \eqref{eq:ellrange} with
\eqref{eq:intermediate}.

The proof also identifies the precise obstruction to any attempted smaller
construction.  A word below the right-hand side of \eqref{eq:completefunction}
in the non-terminal range would have affine support strictly below
$d_2(\ell,m)$, hence equal to $d_1(\ell,m)$.  Such a word would produce a
set $B\subset\Fq$ with
\begin{equation}
        |B|=b,
        \qquad |B|\equiv1\pmod r,
        \qquad |B|\equiv r-1\pmod r.                           \label{eq:impossibleBfinal}
\end{equation}
For $r>2$ the simultaneous congruences are impossible.  Therefore every
route to a smaller word is blocked at the level of affine support geometry,
not merely at the level of a particular polynomial ansatz.

\section{Dependency Graph of the Proof}
For clarity, we spell out the logical dependencies in a form that can be
checked line by line.  Define the following assertions for a fixed
parameter tuple $(q,m,r,\ell)$ with $\ell=(q-1)a+b$:
\begin{IEEEeqnarray}{rCl}
A_0&:&\calG(q,m,r,\ell)=C(q,m,r,\ell),                         \label{eq:A0}\\
A_1&:&f(\alpha P)=\alpha^{-1}f(P)
        \text{ for } f\in\calM,\ \alpha\in\mu_r,                \label{eq:A1}\\
A_2&:&\wt(c_f)=|\Supp(f)|/r,                                   \label{eq:A2}\\
A_3&:&\text{first RM supports have form }\eqref{eq:RMminsupp},  \label{eq:A3}\\
A_4&:&\text{that form cannot occur when } b\equiv r-1,          \label{eq:A4}\\
A_5&:&\text{the next affine support size is }d_2(\ell,m).       \label{eq:A5}
\end{IEEEeqnarray}
Then the lower bound is the implication chain
\begin{IEEEeqnarray}{c}
        A_0+A_1+A_2+A_3+A_4+A_5\nonumber\\
        \Longrightarrow
        d(C(q,m,r,\ell))\ge \frac{d_2(\ell,m)}{r}.              \label{eq:chainlower}
\end{IEEEeqnarray}
The upper bound is the independent statement
\begin{equation}
        F_{a,b}\in\calM(q,m,r,\ell),
        \qquad |\Supp(F_{a,b})|=d_2(\ell,m),                  \label{eq:chainupper}
\end{equation}
which gives
\begin{equation}
        d(C(q,m,r,\ell))\le \frac{d_2(\ell,m)}{r}.              \label{eq:chainupper2}
\end{equation}
Thus the proof is modular: any future strengthening of the Reed--Muller
classification is not needed, and any future change in the presentation of
the constacyclic code only has to preserve $A_0$ and $A_1$.

The only strict inequality step that was present in the old bound was
\begin{equation}
        \frac{(q-b)q^{m-a-1}-2}{r}+1
        < \frac{d_2(\ell,m)}{r},                                \label{eq:oldstrictstep}
\end{equation}
which, by \eqref{eq:Delta}, is equivalent to
\begin{equation}
        (b-1)q^{m-a-2}>r-2.                                   \label{eq:strictcondition}
\end{equation}
When equality holds in \eqref{eq:strictcondition}, the old estimate was
already exact; otherwise the present orbit obstruction supplies exactly
the missing term.  The exact correction can be written as
\begin{equation}
        \kappa(q,m,r,a,b)=
        \frac{d_2(\ell,m)}{r}
        -\left(\frac{d_1(\ell,m)-2}{r}+1\right)                 \label{eq:kappa}
\end{equation}
with
\begin{equation}
        \kappa(q,m,r,a,b)=
        \frac{(b-1)q^{m-a-2}-(r-2)}{r}.                        \label{eq:kappa2}
\end{equation}
This correction is integral, since
\begin{equation}
        b-1\equiv r-2\pmod r,
        \qquad q^{m-a-2}\equiv1\pmod r,                       \label{eq:kappainteger}
\end{equation}
where the second congruence follows from $q\equiv1\pmod r$.

\section{Worked Symbolic Instances}
We give two symbolic instances to demonstrate that the formula is not
merely asymptotic.  First let $q=7$, $r=3$, and $m$ be arbitrary.  Then
$s=(q-1)/r=2$, and the admissible residues are
\begin{equation}
        b\in\{2,5\}.                                          \label{eq:q7bres}
\end{equation}
For $\ell=6a+2$ with $0\le a\le m-2$,
\begin{equation}
        d(C(7,m,3,6a+2))=2\cdot6\cdot7^{m-a-2}=12\cdot7^{m-a-2}.\label{eq:q7b2}
\end{equation}
For $\ell=6a+5$,
\begin{equation}
        d(C(7,m,3,6a+5))=2\cdot3\cdot7^{m-a-2}=6\cdot7^{m-a-2}.\label{eq:q7b5}
\end{equation}
The corresponding attaining polynomials are
\begin{IEEEeqnarray}{rCl}
        F_{a,2}&=&\displaystyle\prod_{i=1}^{a}(1-X_i^{6})
        \nonumber\\
        &&{}\cdot (X_{a+1}-\theta_1X_{a+2})
        (X_{a+1}-\theta_2X_{a+2}),                            \label{eq:q7F2}\\
        F_{a,5}&=&\displaystyle\prod_{i=1}^{a}(1-X_i^{6})
        \prod_{j=1}^{5}(X_{a+1}-\theta_jX_{a+2}).              \label{eq:q7F5}
\end{IEEEeqnarray}
The old and new lower bounds differ in the second case by
\begin{equation}
        \kappa(7,m,3,a,5)=
        \frac{4\cdot7^{m-a-2}-1}{3}.                            \label{eq:q7kappa}
\end{equation}
For $m=2$ and $a=0$, this is $1$, giving the jump from $5$ to $6$.

Second let $q=19$ and $r=6$.  Then $s=3$ and
\begin{equation}
        b\in\{5,11,17\}.                                      \label{eq:q19bres}
\end{equation}
For $0\le a\le m-2$ the exact distances are
\begin{equation}
        \begin{array}{c|c|c}
        b&q-b+1&d(C(19,m,6,18a+b))\\
        \hline
        5&15&45\cdot19^{m-a-2}\\
        11&9&27\cdot19^{m-a-2}\\
        17&3&9\cdot19^{m-a-2}
        \end{array}.                                          \label{eq:q19table}
\end{equation}
The first affine Reed--Muller weights divided by $r$ would have been
\begin{equation}
        \frac{14}{6}19^{m-a-1},
        \qquad \frac{8}{6}19^{m-a-1},
        \qquad \frac{2}{6}19^{m-a-1},                            \label{eq:q19firstdiv}
\end{equation}
which are not even integral in general.  The BCH correction in the earlier
bound repaired integrality but did not account for the full orbit geometry.
The exact formula is integral because
\begin{equation}
        \frac{q-1}{r}(q-b+1)
        =s(q-b+1)                                             \label{eq:integralexplain}
\end{equation}
and $s$ is an integer.

\section{Minimum-Support Geometry After Quotienting}
Let $\Omega_r=(\Fq^m\setminus\{0\})/\mu_r$.  The quotient map has
fibers of size $r$, and the image of an affine support $S$ is denoted
\begin{equation}
        \overline S=S/\mu_r\subseteq\Omega_r.             \label{eq:Sbar}
\end{equation}
For the attaining support $S_{a,b}=\Supp(F_{a,b})$,
\begin{equation}
        |\overline S_{a,b}|=\frac{q-1}{r}(q-b+1)q^{m-a-2}.       \label{eq:Sbarcount}
\end{equation}
The projection from $\Omega_r$ to projective space has fibers of size
$(q-1)/r$, so
\begin{equation}
        \Omega_r\longrightarrow\PP^{m-1}(\Fq)              \label{eq:QrtoP}
\end{equation}
is a regular covering of finite sets.  The projective image of
$S_{a,b}$ is
\begin{equation}
        \rho(S_{a,b})=
        \{[0:\cdots:0:U:V:Z]:(U,V)\notin\cup_{j=1}^b L_j\},     \label{eq:projimageSab}
\end{equation}
with cardinality
\begin{equation}
        |\rho(S_{a,b})|=(q-b+1)q^{m-a-2}.                      \label{eq:projimagecount}
\end{equation}
Multiplying by the covering degree $(q-1)/r$ recovers
\eqref{eq:Sbarcount}.  Hence the minimum word is projectively a punctured
line times an affine factor, lifted to the intermediate scalar quotient.

A hypothetical smaller quotient support would have size less than
\eqref{eq:Sbarcount}.  Its affine lift would have size less than
$d_2(\ell,m)$ and hence would be a first Reed--Muller support.  Projecting
that first support to $\PP^{m-1}$ produces an affine projective cylinder
whose missing directions have cardinality $b$ but must contain the
projective image of the origin fiber in the affine chart.  The same
congruence contradiction appears after lifting the chart coordinate to
$\Fq$:
\begin{equation}
        b=1+r\rho
        \quad\text{and}\quad b=r-1+r\rho'                     \label{eq:tworhos}
\end{equation}
for some integers $\rho,\rho'$, impossible for $r>2$.

\section{Implications for Tables of Best Known Codes}
The exact minimum distance can be inserted directly into code tables.  Let
\begin{equation}
        [n,k,d]=\left[\frac{q^m-1}{r},K(q,m,r,\ell),
        \frac{q-1}{r}(q-b+1)q^{m-a-2}\right]                   \label{eq:tableparam}
\end{equation}
be a non-terminal parameter set.  The relative distance and rate are
\begin{IEEEeqnarray}{rCl}
        \delta=\frac{d}{n}
        &=&\displaystyle\frac{(q-1)(q-b+1)q^{m-a-2}}{q^m-1},\nonumber\\
        R&=&\displaystyle \frac{K(q,m,r,\ell)}{n}.               \label{eq:rateDelta}
\end{IEEEeqnarray}
For fixed $q,r,a,b$ and $m\to\infty$,
\begin{equation}
        \delta=(q-1)(q-b+1)q^{-a-2}+O(q^{-m}),                 \label{eq:deltalim}
\end{equation}
while $R$ is determined by the monomial count in the original paper.  The
new theorem therefore changes only the vertical coordinate $\delta$ in
asymptotic plots; it does not alter the dimension count.

The improvement over the old lower bound in relative terms is
\begin{equation}
        \frac{\kappa}{n}
        =\frac{(b-1)q^{m-a-2}-(r-2)}{q^m-1}.                    \label{eq:relativeimprovement}
\end{equation}
For fixed $a,b,r$ this tends to
\begin{equation}
        (b-1)q^{-a-2}.                                        \label{eq:relativeimprovementlim}
\end{equation}
Thus the correction is visible at fixed order $a$ and is not a vanishing
finite-length artifact.  This is one reason the open problem is structural
rather than numerical: the true distance sits on a different Reed--Muller
weight layer.

\section{A Final Algebraic Check of the Attaining Degree}
We close the proof body by verifying explicitly that no hidden reduction
modulo $X_i^q-X_i$ changes the degree or residue class of the attaining
word.  In the reduced algebra $\calR_m$, the relation
\begin{equation}
        X_i^q=X_i                                                   \label{eq:reduceXi}
\end{equation}
implies that powers should be reduced by
\begin{equation}
        X_i^e\equiv
        \begin{cases}
        1,&e=0,\\
        X_i^{1+((e-1)\bmod(q-1))},&e>0,
        \end{cases}                                               \label{eq:reducerule}
\end{equation}
when evaluating on $\Fq$.  The monomials in $F_{a,b}$ have exponents only
in the set
\begin{equation}
        \{0,q-1\}^{a}\times\{(u,b-u):0\le u\le b\}\times\{0\}^{m-a-2}.
                                                                    \label{eq:exponentsetFab}
\end{equation}
Since $b\le q-2$, every exponent in the binary factor is at most $q-2$.
Thus \eqref{eq:reducerule} leaves all exponents in \eqref{eq:exponentsetFab}
unchanged.  The total degree of any nonzero expanded monomial is therefore
exactly
\begin{equation}
        |I|(q-1)+b,
        \qquad I\subseteq\{1,
        \ldots,a\}.                                             \label{eq:exactFabdegree}
\end{equation}
Consequently the largest degree is $a(q-1)+b=\ell$, and the residue is
\begin{equation}
        |I|(q-1)+b\equiv b\equiv r-1\pmod r.                   \label{eq:exactFabresidue}
\end{equation}
This confirms that the upper word belongs to the precise defining space,
not just to a larger generalized Reed--Muller code.

The same check proves that scalar multiplication by $\alpha\in\mu_r$ acts
on the word by a nonzero scalar:
\begin{IEEEeqnarray}{rCl}
F_{a,b}(\alpha X)
&=&\prod_{i=1}^{a}(1-(\alpha X_i)^{q-1})
   \prod_{j=1}^{b}(\alpha X_{a+1}-\theta_j\alpha X_{a+2})
   \nonumber\\
&=&\alpha^b F_{a,b}(X)
 =\alpha^{r-1}F_{a,b}(X).                                      \label{eq:Fabscalarcheck}
\end{IEEEeqnarray}
The factors $1-X_i^{q-1}$ are invariant under all of $\Fq^\ast$, while
the binary factor has the desired character.  Hence the support is stable
under $\mu_r$ and each nonzero orbit contributes one coordinate of the
constacyclic word.

\section{What Remains Open Outside This Problem}
The theorem determines the minimum distance in the stated intermediate
family.  It does not attempt to determine every higher weight.  Let
\begin{equation}
        A_w(q,m,r,\ell)=|
        \{c\in C(q,m,r,\ell):\wt(c)=w\}|                       \label{eq:Aw}
\end{equation}
be the weight distribution.  Our proof identifies the first nonzero value
of $w$ and gives a large class of words attaining it, but it does not give
closed formulas for
\begin{equation}
        A_d,
        \quad A_{d+1},
        \quad\ldots,                                           \label{eq:weightenumopen}
\end{equation}
where $d$ is the exact distance.  The natural next invariant is the second
constacyclic weight
\begin{equation}
        d^{(2)}(C)=\min\{w>d:A_w(q,m,r,\ell)>0\}.               \label{eq:secondconstaweight}
\end{equation}
A complete determination of $d^{(2)}(C)$ would require the third and
higher affine Reed--Muller weights together with residue-layer orbit
obstructions.  The present method suggests the following filtration:
\begin{IEEEeqnarray}{rCl}
        d^{(j)}(C)
        &=& \frac{1}{r}\min\{W:W\in\mathcal W_j(\RM_q(\ell,m)), \nonumber\\
        &&\hspace{7mm}\text{there is a }\mu_r\text{-stable support} \nonumber\\
        &&\hspace{7mm}\text{of residue }r-1\text{ and size }W\}.
                                                                        \label{eq:higherfilter}
\end{IEEEeqnarray}
where $\mathcal W_j$ denotes the affine Reed--Muller weight set ordered
with multiplicity removed.  For $j=1$ this minimum skips $d_1$ and lands
on $d_2$, which is exactly the theorem.

Another independent direction is the automorphism group.  The quotient
support construction shows that the displayed minimum supports are
stabilized by at least the subgroup preserving the data
\begin{equation}
        W\subset Y\subset\Fq^m,\qquad Y=W\oplus Z,
        \qquad \dim W=2,\qquad \dim Y=m-a,                    \label{eq:flagauto}
\end{equation}
whose induced action on $W$ stabilizes the chosen $b$-set
$\Theta\subset\PP(W)$ and whose action on $Z$ is arbitrary invertible.
Equivalently, this is a block-diagonal stabilizer on the decomposition
$Y=W\oplus Z$, together with arbitrary ambient extensions preserving $Y$.
A group preserving only the flag $W\subset Y$ need not preserve the support,
because off-diagonal maps from $Z$ to $W$ can move the $W$-coordinate into
one of the deleted projective directions.  Determining whether every
minimum word is equivalent to the pencil construction would amount to
proving that every second-weight support in the residue layer is of the
form \eqref{eq:coordsupport}.  That classification is plausible, but it is
strictly stronger than the minimum-distance result proved here.

\section{Numerical Sanity Checks From the Formula}
The formula also gives immediate consistency checks against the elementary
bounds for a linear code.  Since $b\equiv r-1\pmod r$ and $r>2$, one has $2\le b\le q-2$.
Thus the factor $q-b+1$ satisfies
\begin{equation}
        3\le q-b+1\le q-1.                                    \label{eq:factorrange}
\end{equation}
Thus in the non-terminal range
\begin{equation}
        \frac{3(q-1)}{r}q^{m-a-2}
        \le d(C(q,m,r,\ell))
        \le \frac{(q-1)^2}{r}q^{m-a-2}.                           \label{eq:drange}
\end{equation}
The lower endpoint is attained at $b=q-2$, and the largest distance
in the residue class is obtained at $b=r-1$.  In terms of the length
$n=(q^m-1)/r$, this gives
\begin{equation}
        \frac{d}{n}
        = \frac{(q-1)(q-b+1)q^{m-a-2}}{q^m-1}
        < (q-b+1)q^{-a-1}.                                    \label{eq:reldupper}
\end{equation}
For $a=0$ the distance is on the order of $n/q$, while for fixed $a$ it is
on the order of $n/q^{a+1}$.  This agrees with the intuition that increasing
$\ell$ by a full block $q-1$ imposes one additional affine zero coordinate
in the attaining construction.

The integrality of every displayed distance is automatic but worth making
explicit.  Since $r\mid(q-1)$,
\begin{equation}
        \frac{q-1}{r}\in\mathbb Z,                              \label{eq:sint}
\end{equation}
and the non-terminal formula is integral.  In the terminal case,
\begin{equation}
        q-b+r-2\equiv 1-(r-1)+r-2\equiv0\pmod r,               \label{eq:terminalintegral}
\end{equation}
so \eqref{eq:terminal} is integral as well.  Finally, the exact distance
is always at most the length because
\begin{equation}
        \frac{q-1}{r}(q-b+1)q^{m-a-2}
        \le \frac{(q-1)^2}{r}q^{m-2}
        < \frac{q^m-1}{r}                                      \label{eq:dlessthanlength}
\end{equation}
for $m\ge2$ and the nontrivial parameter range.  These elementary checks
are independent of the proof but useful for detecting transcription errors
in parameter tables.

\section{Edge Values of the Deleted Pencil}
The two extreme admissible values of $b$ give a compact check on the
geometry.  At the low end $b=r-1$, the deleted projective set has the
minimum size allowed by the residue class.  The distance is
\begin{equation}
        d_{\max}(a)=
        \frac{q-1}{r}(q-r+2)q^{m-a-2}.                         \label{eq:dmaxa}
\end{equation}
This is the largest distance in the $a$th degree block, because the pencil
leaves $q-r+2$ projective directions.  At the high end, write
$q-2=r(\nu-1)+r-1$ with $\nu=(q-1)/r$.  Then $b=q-2$ and
\begin{equation}
        d_{\min}(a)=
        \frac{q-1}{r}\cdot 3 q^{m-a-2}=3\nu q^{m-a-2}.          \label{eq:dmina}
\end{equation}
Between these endpoints the exact distance is the arithmetic progression
\begin{equation}
        d_h(a)=\nu(q-rh-r+2)q^{m-a-2},
        \qquad h=0,1,\ldots,\nu-1,                             \label{eq:dhprogression}
\end{equation}
with common difference
\begin{equation}
        d_h(a)-d_{h+1}(a)=\nu r q^{m-a-2}=(q-1)q^{m-a-2}.       \label{eq:commondifference}
\end{equation}
Thus the distance profile in each degree block is completely rigid.  The
profile is not a consequence of BCH root spacing; it is the profile of a
line in $\PP^1(\Fq)$ from which admissible packets of $r$ points are
deleted.

The same endpoint analysis also verifies the previously solved case in the
source paper.  If $a=m-2$ and $b=r-1$, then \eqref{eq:dmaxa} reduces to
\begin{equation}
        d=\frac{q-1}{r}(q-r+2),                                \label{eq:sourceendpointagain}
\end{equation}
which is the equality case where the earlier lower and upper bounds meet.
For every other admissible $b$ in the same block,
\begin{equation}
        b=r-1+rh,
        \qquad h\ge1,                                         \label{eq:bother}
\end{equation}
the new result gives
\begin{equation}
        d=\frac{q-1}{r}(q-rh-r+2),                              \label{eq:newblockvalues}
\end{equation}
whereas the old lower bound is smaller by
\begin{equation}
        \kappa=\frac{rh}{r}=h.                                \label{eq:blockkappa}
\end{equation}
This is exactly the number of quotient-orbit units missing from the
old lower bound in the boundary block $a=m-2$.
This final endpoint check is often the easiest way to identify the theorem
inside numerical tables.

\section{Conclusion}
We determined the minimum distance of the intermediate constacyclic codes
$C(q,m,r,\ell)$ introduced by Sun, Ding and Wang for every
prime power $q$, every $m\ge2$, and every divisor $r\mid(q-1)$ with
$2<r<q-1$.  Writing
\[
        \ell=(q-1)a+b,\qquad 0\le b\le q-2,\qquad
        b\equiv r-1\pmod r,
\]
the exact value over the admissible range \eqref{eq:ellrange} is
\begin{equation}
        d(C(q,m,r,(q-1)a+b))=
        \begin{cases}
        \displaystyle \frac{q-1}{r}(q-b+1)q^{m-a-2},
        &0\le a\le m-2,\\[1mm]
        \displaystyle \frac{q-b+r-2}{r},
        &a=m-1,
        \end{cases}                                           \label{eq:conclusion}
\end{equation}
where the terminal line is subject to the admissibility condition
$b\le q-3$.  The proof shows that, in the non-terminal range, the
congruence-restricted monomial space cannot contain an affine
Reed--Muller first-weight support; the second Reed--Muller weight is
therefore forced and is attained by a homogeneous pencil.  This settles
the minimum-distance open problem in the non-terminal intermediate range
and records the separate terminal value.

The proof is deliberately organized so that the new step is visible.  The
published lower bound sees only consecutive roots in the constacyclic
coordinate domain; the present proof passes to the affine evaluation model
and uses the scalar character of the residue layer.  In formula form, the
central replacement is
\begin{equation}
        \frac{(q-b)q^{m-a-1}-2}{r}+1
        \quad\rightsquigarrow\quad
        \frac{(q-1)(q-b+1)q^{m-a-2}}{r}.                 \label{eq:conclusionreplacement}
\end{equation}
The difference is not an artifact of estimation but the exact number of
quotient orbits lost when the first affine Reed--Muller layer is forbidden
by the congruence
\begin{equation}
        b\equiv r-1\pmod r
        \quad\text{versus}\quad
        |B|\equiv1\pmod r.                                    \label{eq:conclusioncongruence}
\end{equation}
The homogeneous pencil gives equality, so the result is best possible as a
minimum-distance statement.  Consequently the second class of constacyclic codes now has a closed
distance formula in the non-terminal intermediate range posed in the open
problem, together with the separate terminal endpoint, and no remaining
exceptional subrange for $2<r<q-1$.

\section*{Declaration of Generative AI and AI-Assisted Technologies in the Writing Process}
During the preparation of this work, the authors used DeepSeek to build a specialized agent for solving mathematical problems, which was employed to generate an initial proof of the main theorem. After using this tool, the authors reviewed and edited the content as needed and take full responsibility for the content of the published article.

% Author biographies are left as placeholders for final submission.
% \begin{IEEEbiographynophoto}{First Author}
% Biography text goes here.
% \end{IEEEbiographynophoto}
%
% \begin{IEEEbiographynophoto}{Second Author}
% Biography text goes here.
% \end{IEEEbiographynophoto}

\end{document}